\documentclass[10pt]{article}
\usepackage[margin=1in]{geometry}

\usepackage{amssymb}
\usepackage{amsmath}
 \usepackage{amsthm}

\usepackage{hyperref}
\usepackage{subcaption}
\usepackage{booktabs}
\usepackage{tikz,multirow}
\newcommand{\de}{\partial}
\newcommand{\sign}{\text{sign}}
\usepackage{algorithm}
\usepackage{algpseudocode}

\newtheorem{proposition}{Proposition}
\newtheorem{definition}{Definition}
\usepackage{cite}
\usepackage{microtype}
\usepackage{authblk}   
\title{Second-order multilane traffic flow models: from the microscopic to the macroscopic scale}

\author[1,*]{Matteo Piu}
\author[2]{Giuseppe Visconti}
\author[2]{Gabriella Puppo}

\affil[1]{University of Verona, Strada le Grazie 15, 37134 Verona, Italy\\
  \texttt{matteo.piu@univr.it}}
\affil[2]{Sapienza -- University of Rome, Piazzale Aldo Moro 5, 00185 Rome, Italy\\
  \texttt{\{giuseppe.visconti,gabriella.puppo\}@uniroma1.it}}

\affil[*]{Corresponding author.}

\date{} 

\begin{document}

\maketitle

\begin{abstract}
\noindent This study addresses multilane vehicular traffic modelling, focusing on the rigorous transition from microscopic (individual vehicle-based) to macroscopic (aggregate flow-based) descriptions. While previous research on multilane traffic has largely focused on first-order models, we derive two novel multilane second-order macroscopic models by applying a microscopic-to-macroscopic limit to the multilane Bando–Follow-the-Leader model. These two models diverge fundamentally in their closure relations during the scaling limit: Model 1 directly projects the velocity dynamics, whereas Model 2 is built upon the heuristic conservation of the generalized momentum across lanes. Both models incorporate lane-changing dynamics through source terms in a hyperbolic system of balance laws, yet their structural differences lead to distinct relaxation limits, with Model 2 naturally relaxing to the first-order macroscopic multilane model recently derived by the authors, and Model 1 yielding a novel, non-standard first-order system. We propose several numerical experiments showing that the models can reproduce complex traffic phenomena, including congestion propagation, non-equilibrium effects, capacity drops, and asymmetric lane usage. Leveraging experimental datasets from real-world highways, we further construct lane-specific empirical fundamental diagrams and compare them with their simulated counterparts. The results demonstrate that the structural assumptions of Model 2 provide a superior quantitative fit with empirical data, faithfully capturing critical density values, traffic scattering, and characteristic lane-dependent patterns, thus offering a robust and generalizable tool for realistic traffic flow analysis.
\end{abstract}

\noindent\textbf{2020 Mathematics Subject Classification.} 76A30, 35L65, 35L40, 35L60.

\noindent\textbf{Keywords and phrases:} traffic flow, multilane, lane changing dynamics, microscopic-to-macroscopic, second-order models, balance laws, relaxation limits, fundamental diagrams.


\section{Introduction}
The modeling of traffic flow is fundamental for understanding and managing vehicular mobility in modern transportation networks. As urbanization and population continue to grow, effective traffic models are crucial for the planning, optimization, and control of traffic systems. In the context of increasing road congestion, environmental concerns, and the push toward sustainable mobility, accurate and efficient traffic models play a central role in traffic planning with the aim of reducing travel times, mitigating emissions, and supporting the integration of emerging technologies such as intelligent transportation systems and autonomous vehicles.

From a mathematical standpoint, three primary approaches address this problem, each corresponding to a different scale of representation: microscopic (based on individual vehicles), mesoscopic (based on probability distributions  of vehicles velocities), and macroscopic (where aggregate quantities such as density or local mean velocity are evolved)  \cite{libropuppo,review:Hoogendoorn,hoogendoorn2001state}. Microscopic models describe the dynamics of individual vehicles, typically through systems of ordinary differential equations. Mesoscopic models, inspired by kinetic theory, provide a statistical description of vehicle densities and velocities, offering a compromise between detail and computational tractability. Macroscopic models, often based on systems of partial differential equations, treat traffic flow as a continuum and describe the evolution of aggregated quantities such as vehicle density and mean velocity. The connection between these models lies in their scale of representation and in the limiting procedures that allow transitions between them.  Theoretical frameworks  rigorously establish these relationships~\cite{HPRV20}, bridging detailed vehicle behavior and large-scale traffic phenomena.

Here, we focus on  microscopic and  macroscopic approaches, their application in multilane traffic scenarios, and the study of the transition between the two scales. Studying the connections between microscopic and macroscopic traffic models is advantageous for at least two primary reasons. First, it enables the design of macroscopic models that faithfully reproduce the essential dynamics observed at the microscopic level, thereby leading to more accurate and reliable continuum descriptions. Second, it allows leveraging the significant computational efficiency of macroscopic models, particularly suitable for large-scale simulations and real-time control, while retaining a rigorous foundation rooted in individual vehicle interactions. Moreover, since macroscopic models operate on aggregate quantities, they reduce the dependence on detailed initial data, which is often unavailable in practical scenarios.

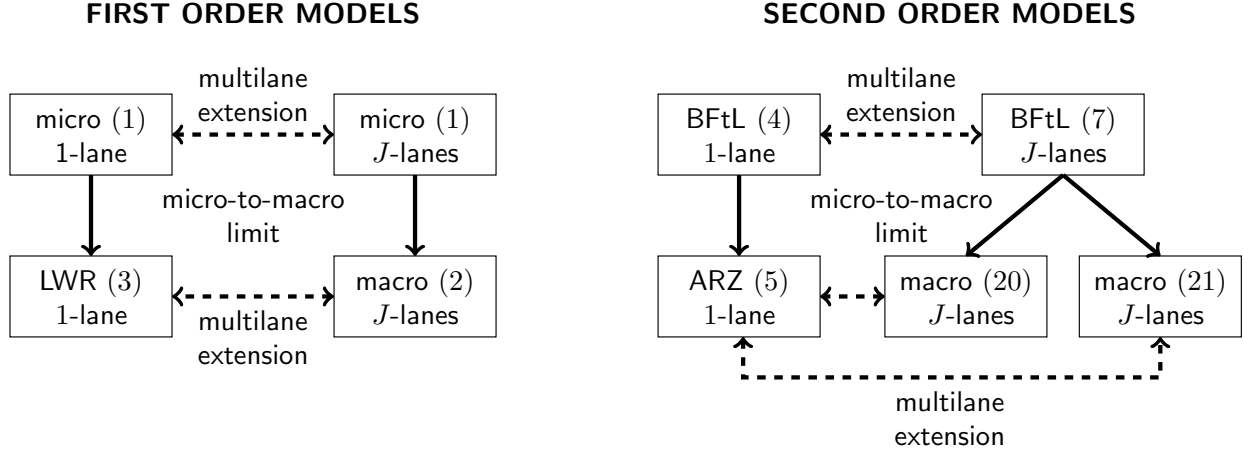
\begin{figure}
\centering
\resizebox{\textwidth}{!}{
\begin{tikzpicture}
\draw[draw=black] (0,0) rectangle (2,1);
\node[align=center] at (1,0.5) {\textsf{micro \eqref{micro_1_ord}}\\ \textsf{1-lane}};

\draw[draw=black] (4,0) rectangle (6,1);
\node[align=center] at (5,0.5) {\textsf{micro \eqref{micro_1_ord}}\\ \textsf{$J$-lanes}};

\draw[draw=black] (0,-2) rectangle (2,-1);
\node[align=center] at (1,-1.5) {\textsf{LWR} \eqref{lwr}\\ \textsf{$1$-lane}};

\draw[draw=black] (4,-2) rectangle (6,-1);
\node[align=center] at (5,-1.5) {\textsf{macro \eqref{macro_1_ord}}\\ \textsf{$J$-lanes}};

\draw[draw=black,<->,line width=1.5pt,dashed](2,0.5)--(4,0.5);
\node[align=center] at (3,1) {\textsf{multilane}\\ \textsf{extension}};

\draw[draw=black,->,line width=1.5pt](1,0)--(1,-1);
\draw[draw=black,->,line width=1.5pt](5,0)--(5,-1);
\node[align=center] at (3,-0.5){\textsf{micro-to-macro}\\ \textsf{limit}};

\draw[draw=black,<->,line width=1.5pt,dashed](4,-1.5)--(2,-1.5);
\node[align=center] at (3,-2) {\textsf{multilane}\\ \textsf{extension}};

\draw[draw=black] (0+8,0) rectangle (2+8,1);
\node[align=center] at (1+8,0.5) {\textsf{BFtL \eqref{KBFTL}}\\ \textsf{$1$-lane}};

\draw[draw=black] (4+8,0) rectangle (6+8,1);
\node[align=center] at (5+8,0.5) {\textsf{BFtL} \eqref{bftl_nlane}\\ \textsf{$J$-lanes}};

\draw[draw=black] (0+8,-2) rectangle (2+8,-1);
\node[align=center] at (1+8,-1.5) {\textsf{ARZ} \eqref{eq:arz} \\ \textsf{$1$-lane}};

\draw[draw=black] (2.8+8,-2) rectangle (4.8+8,-1);
\node[align=center] at (3.8+8,-1.5) {\textsf{macro \eqref{eq:model1:noncons}}\\\ \textsf{$J$-lanes}};

\draw[draw=black] (5.2+8,-2) rectangle (7.2+8,-1);
\node[align=center] at (6.2+8,-1.5) {\textsf{macro \eqref{eq:model2:noncons}}\\ \textsf{$J$-lanes}};
\draw[draw=black,<->,line width=1.5pt,dashed](2+8,0.5)--(4+8,0.5);
\node[align=center] at (3+8,1) {\textsf{multilane}\\ \textsf{extension}};

\draw[draw=black,->,line width=1.5pt](1+8,0)--(1+8,-1);
\draw[draw=black,->,line width=1.5pt](5+8,0)--(3.8+8,-1);
\draw[draw=black,->,line width=1.5pt](5+8,0)--(6.2+8,-1);
\node[align=center] at (3+8,-0.5) {\textsf{micro-to-macro}\\ \textsf{limit}};

\draw[draw=black,<->,line width=1.5pt,dashed](2.8+8,-1.5)--(2+8,-1.5);
\draw[draw=black,<->,line width=1.5pt,dashed](2.8+6.25,-2.0)--(2.8+6.25,-2.5)--(5.2+9,-2.5)--(5.2+9,-2.0);
\node[align=center] at (3.6+8,-3.0) {\textsf{multilane}\\ \textsf{extension}};

\node[align=center] at (3,2) {\textbf{\textsf{FIRST-ORDER MODELS}}};
\node[align=center] at (3.6+8,2) {\textbf{\textsf{SECOND-ORDER MODELS}}};
\end{tikzpicture}}
\caption{Links between models. The left panel illustrates the connections between first-order models, including the transition from microscopic (\emph{micro}) to macroscopic (\emph{macro}) scales and their multilane extensions. The right panel depicts the second-order models and their respective transitions and extensions. Dashed arrows represent multilane extensions, while solid arrows indicate the micro-to-macro limit.}
\label{fig:relazioni}
\end{figure}

\paragraph{State of the art} The modeling of multilane traffic flow presents additional challenges compared to single-lane scenarios, due to interactions such as lane-changing maneuvers and variable speeds across lanes. Several multilane traffic models have been proposed to capture these complexities. For example, Klar and Wegener~\cite{KlarWegener98:I,KlarWegener98:II} present a multilane traffic model that begins with the modeling of microscopic interaction rules along the road and across the lanes. From these, they derive an Enskog-like kinetic model and compute a fluid-dynamic limit. Greenberg, Klar and Rascle~\cite{GreenbergKlarRascle} describe an aggregate lane approach where distinct equilibrium relationships are considered to model multilane effects. More recently, Holden and Risebro~\cite{HoldenRisebro} extend the first-order LWR model to multiple lanes by modeling lane-changing based on velocity differences. Their model is described by a coupled system of hyperbolic balance laws. Instead, Song and Karni~\cite{SongKarni} focus on the derivation of a second-order model for multilane traffic. Models involving two coupled evolution equations, one for the density and one for  velocity, are commonly referred to as second-order models in the literature. These models differ from first-order ones in that velocity is treated as an independent variable, separate from density, and governed by its own evolution equation. As a result, second-order models form $2 \times 2$ hyperbolic systems. 

In a previous paper~\cite{piu2}, we derived a first-order macroscopic multilane model as the limit of a first-order hybrid microscopic model with lane change (LC) dynamics governed by two lane-changing rules: incentive and safety. The incentive criterion triggers a lane change if a vehicle in lane $j$ would travel faster in lane $j'$, while the safety criterion ensures that the maneuver is admissible only when sufficient space is available in the new lane. The microscopic model is obtained as the multilane extension of a single-lane model 
and is of the type
\begin{equation} \label{micro_1_ord}
\begin{cases}
\dot x_n = \tilde{V}_j(\Delta x_n^j), \\
\text{Change lane to lane $j'$ if:} \\
\quad \tilde{V}_{j'}(\Delta x_n^{j'}) > \dot{x}_n \quad \text{(incentive criterion)}, \\
\quad \Delta x^{j'}_n > l + d_s \ \text{and} \ x_n - x_{p_n^{j'}} > l + d_s \quad \text{(safety criterion)}.
\end{cases}
\end{equation}
for $n \in I_j$ and $j=1,\dots,J$.
Here, \( J \) denotes the number of lanes, $I_j$ is the set of vehicles in lane $j$, $x_n\in \mathbb{R}$ is the position of the $n$th vehicle and refers to its rear (i.e., to its back bumper), $\Delta x_n^j$ is the headway between this vehicle and its leader in lane $j$, and $p_n^{j'}$ represents the potential follower of vehicle $n$ in the target lane $j'$, as in Fig.~\ref{strada}. The parameter \( l>0 \) represents the vehicle length, assumed identical for all vehicles (single homogeneous population), while \( d_s>0 \) is the minimum safety distance. The functions $\tilde{V}_j(\cdot)$ prescribe the desired or equilibrium velocity as a function of the headway.
 The macroscopic model derived from model \eqref{micro_1_ord} consists in a system of $J$ balance laws, one per lane, given by \cite{piu2}
\begin{equation}\label{macro_1_ord}
\partial_t \rho_j + \partial_x f_j(\rho_j) = \nu S_j, \quad j=1,\dots,J.
\end{equation}
Here, \(\rho_j = \rho_j(x, t)\) is the density of vehicles in lane \( j \), where \(\rho_j: \mathbb{R} \times [0, +\infty) \to [0, \rho^{\text{max}}]\). The flux function is given by the relation \( f_j(\rho_j) = \rho_j V_j(\rho_j) \), where \( V_j: [0, \rho^{\text{max}}] \to [0, V_j^{\text{max}}] \) denotes the equilibrium speed, which is assumed to be a strictly decreasing function of the density, with \( V_j(0) = V_j^{\text{max}} \) and \( V_j(\rho^{\text{max}}) = 0 \). The source term \( S_j=S_j(\rho_{j-1},\rho_{j},\rho_{j+1}) \) consists of two loss and gain terms that describe the mass exchange from/to lane \( j \) to/from the adjacent lanes \( j \pm 1 \), scaled with a parameter $\nu>0$ denoting the lane changing frequency. For the detailed structure of these source terms, we refer to the already cited work \cite{piu2}. Note that, in the case \( J=1 \), i.e., the single-lane scenario, the source terms vanish, and the model \eqref{macro_1_ord} coincides with the well-known LWR model 
\begin{equation}\label{lwr}
\partial_t \rho + \partial_x f(\rho) = 0,
\end{equation}
with $f(\rho)=\rho V(\rho)$, which consists of a scalar conservation law modelling mass conservation. Moreover, the LWR model can be viewed as a many-particle limit of the discrete (microscopic) model of interacting particle systems given by \eqref{micro_1_ord} in the single-lane case \cite{microtomacro1, microtomacro2, microtomacro3, microtomacro4}. The connections between all these first-order models are summarized in the left panel of Fig. \ref{fig:relazioni}.

In this work, we seek to extend the previous results in the context of second-order models. In first-order traffic models, only the density evolves in time, while the velocity is determined instantaneously through a prescribed fundamental relation with the density. In contrast, second-order models include an additional evolution equation for the velocity, allowing the description of non-equilibrium phenomena that are particularly relevant in multilane flows. Compared to first-order models, second-order macroscopic models account for anticipation effects, which describe drivers' ability to adjust their speed not only based on their current spacing, but also based on variations in traffic conditions ahead, effectively reacting to spatial gradients of density or velocity. This leads to more realistic dynamics, especially in congested or multilane traffic scenarios. More specifically, we aim to derive multilane second-order macroscopic models, obtained as the limit of microscopic models in analogy with the approach taken in \cite{piu2}.

To achieve this, we start considering the single-lane microscopic ``Bando--Follow-the-Leader'' (BFtL) model \cite{piu1,gong,stern,seibold} for a set of $N$ vehicles, described by the following system of $2N$ ordinary differential equations (ODEs):
\begin{equation}\label{KBFTL}
\begin{cases}
\dot{x}_n = v_n \\
\dot{v}_n = \alpha (\tilde{V}(\Delta x_n) - v_n) + \beta \dfrac{\Delta v_n}{(\Delta x_n - K)^{\gamma + 1}}
\end{cases} \quad n = 1, \dots, N,
\end{equation}
for $\alpha,\beta>0$ and \( \gamma \geq 0 \). The variables $x_n$ and \( v_n \) represent the position and the velocity of vehicle $n$ respectively, while \( \Delta x_n \) and \( \Delta v_n \) denote the headway and the relative velocity, respectively, between the \( n \)th vehicle and its leading vehicle.  Moreover, \( K\in\{0,l+d_s\} \) is a constant, and if \( K = 0 \), we recover the classical BFtL model \cite{aw2002}. Conversely, if \( K = l + d_s \) we obtain the ``Modified'' BFtL model \cite{degond}. Notably, \( K \) denotes the minimum distance between vehicles. In the classical model, vehicles are allowed to approach each other almost to the point of contact, in the so-called ``bumper-to-bumper'' scenario, whereas in the modified version, the distance \( d_s>0 \) is maintained between vehicles. The interaction between vehicles is characterized by the two terms  in the second equation in~\eqref{KBFTL}. The first term \cite{bando1,bando2} is a relaxation term towards an optimal (desired) velocity. The optimal velocity function \( \tilde{V}: [0, +\infty) \to [0, V^{\text{max}}),\) with \( \tilde{V} \in C^1 \cap L^\infty\), is typically assumed to be a bounded and monotonically increasing function of the headway. It approaches zero for small headways and reaches a maximum value \( {V}^{\text{max}} \) for larger headways. This term is scaled by a parameter \( \alpha > 0 \), which denotes the response speed, i.e., the sensitivity of the driver, with dimensions of the inverse of  time. The second term corresponds to the classical FtL interaction \cite{ftl1}, scaled by a parameter \( \beta > 0 \) with dimensions of length raised to the power of \( \gamma + 1 \) over time. In this interaction, the acceleration of a vehicle is directly proportional to the relative velocity of the interacting vehicles and inversely proportional to their mutual distance raised to the power of \( \gamma + 1 \).

The macroscopic counterpart of the BFtL model \eqref{KBFTL} is given by the Aw-Rascle-Zhang (ARZ) model \cite{arz1,arz2}, as shown with a rigorous derivation in \cite{aw2002,degond,rosini}. The model consists of a $2\times 2$ system of partial differential equations (PDEs) given by \begin{equation}\label{eq:arz}
\begin{cases}
\partial_t \rho + \partial_x (\rho v) = 0, \\
\partial_t v + (v - \rho P'(\rho)) \partial_x v = \alpha \left( V(\rho) - v \right),
\end{cases}
\end{equation} where $\rho = \rho(x,t)$ is the vehicle density and $v = v(x,t)$ is the mean velocity. The function $P(\rho)$ is called the pressure or hesitation function,  which accounts for drivers' reactions to the traffic state ahead and is typically taken as an increasing function of the density. In the original ARZ model \cite{aw2002}, obtained as the limit of the microscopic model \eqref{KBFTL} with $K=0$, the pressure has the form $P(\rho) \sim \beta \rho^\gamma$, with $\beta>0$ and $\gamma\geqslant 0$. As it is well known, selecting such pressure term can result in densities that exceed the maximum allowable value, which poses a challenge for modeling, since the maximum density has a well-defined physical interpretation, as long as the model is based on vehicles of a finite and non zero length. To address this issue, alternative models have been proposed in the literature that account for this constraint \cite{gsom1,gsom2,gsom3}. In particular, we mention the Modified ARZ (MARZ) model in~\cite{degond}, which is obtained as the limit of model \eqref{KBFTL} with $K=l+d_s$. This choice leads to a pressure $P(\rho) \sim \beta ({1}/{\rho} + {1}/{\rho^{\text{max}}})^{-\gamma}$ that ensures that the density remains within its physical limits.

The system described by \eqref{eq:arz} remains strictly hyperbolic when the density is strictly positive, featuring one genuinely nonlinear characteristic field and one that is linearly degenerate. The linearly degenerate field corresponds to the faster eigenvalue, which equals \( v \), and it gives rise to contact discontinuities, where the propagation speed remains continuous and is determined by \( v \), see \cite{arz1} for a deeper analysis. From a numerical perspective, approximating contact discontinuities presents several challenges. For instance, classical conservative schemes, such as Rusanov or Godunov, can generate significant nonphysical waves that degrade the numerical solution over long times. Various techniques have been introduced to address this problem, also in relation to the specific case of the ARZ system \cite{osci1, osci2, osci3}. In this paper, we will focus on a Roe-type wave-propagation scheme, which is an upwind scheme that linearizes the Jacobian matrix at each cell interface, as developed by Roe \cite{roe1, roe2}. This approach has been used also in~\cite{SongKarni}  and it accurately captures contact discontinuities in traffic flow without introducing nonphysical oscillations and wrong intermediate states. 

Moreover system \eqref{eq:arz}  can be rewritten in conservative form by introducing the variable \( y := \rho(v + P(\rho)) \). This quantity is frequently regarded as a form of generalized momentum of traffic. With its introduction, the ARZ model can be reformulated as follows:
\begin{equation}\label{eq:arz_cons}
\partial_t \begin{pmatrix} \rho \\ y \end{pmatrix} + \partial_x \begin{pmatrix} y-\rho P(\rho) \\ \frac{y^2}{\rho}-yP(\rho) \end{pmatrix} = \begin{pmatrix} 0 \\ \alpha \rho \left( V(\rho) - \frac{y}{\rho}-P(\rho) \right) \end{pmatrix}.
\end{equation}

In our approach, we explicitly incorporate the effects of lane changes into the macroscopic model by deriving macroscopic traffic laws directly from microscopic dynamics that already include lane-changing behavior. This micro-to-macro derivation enables us to capture lane change effects consistently within the macroscopic framework, rather than treating lane changes as separate or exogenous phenomena. Thus, the complex multilane interactions and their impact on aggregate traffic flow naturally emerge from the underlying microscopic processes. Conventional macroscopic models that do not include lane change effects fail to capture important multilane traffic phenomena such as lane-specific flow heterogeneity, capacity drops near bottlenecks caused by lateral interactions, and the formation of complex congestion patterns like synchronized flow. Without modeling lane changes, either explicitly or through micro-to-macro derivations, these models cannot represent how  lane changing interactions influence overall traffic dynamics.

Another key aspect in the study of traffic dynamics is given by fundamental diagrams, which relate traffic flow to vehicle density (and sometimes speed). These diagrams play a central role in both theoretical and applied traffic flow analysis by providing a compact macroscopic representation of traffic states and serving as a benchmark for model validation. They are directly obtained from experimental measurements collected via loop detectors, radar, or other sensing technologies. Empirical fundamental diagrams reveal crucial features such as the free-flow regime, the critical density marking the transition to congestion, and the flow reduction under high-density conditions. A traffic model is expected to accurately reproduce these features, capturing not only the general shape of the diagram but also lane-specific differences, possible asymmetries between lanes, and behavior near critical transitions. While first-order models are generally capable of reproducing the main features of fundamental diagrams, such as the overall shape, the capacity drop, and the critical density, second-order models offer an important additional capability: they can reproduce the variability observed in experimental data, particularly the scattering in congested regimes.

\paragraph{Plan and contribution of the paper}  
In this paper, we aim to study the macroscopic limit of the hybrid multilane BFtL microscopic model, with the goal of deriving a purely macroscopic description of multilane traffic dynamics based on second-order models. After this introduction, in Section \ref{sec:micro} we revise and summarize the main aspects of the microscopic model presented in \cite{piu1}. Then, Section \ref{sez:macromod} presents a detailed analysis of the macroscopic limit, explored through two different formulations. The first scaling transition is based on analyzing the variation of microscopic vehicle velocities, $v_n$, in the presence of lane changes. This is done by directly applying the BFtL model equations, ultimately deriving a macroscopic description. Similarly, in the second approach, we consider the macroscopic ARZ model in conservative variables \eqref{eq:arz_cons} and investigate how lane changes influence the temporal evolution of the momentum $y$. In both cases, the derivation relies on a detailed examination of the discrete lane-changing behavior of individual vehicles, leading to the formulation of appropriate interaction terms that account for lane-change dynamics. Crucially, these two models differ fundamentally in their closure relations during the micro-to-macro transition. While Model 1 directly projects the velocity dynamics, Model 2 is built upon the heuristic conservation of the generalized momentum $y$ across lanes. This structural divergence significantly impacts their macroscopic properties, most notably their relaxation limits: as we will detail, the relaxation of Model 1 yields a novel, non-standard first-order system, whereas Model 2 naturally relaxes to model \eqref{macro_1_ord}, previously introduced in \cite{piu2}.  The analysis results in the derivation of two novel second-order macroscopic multilane models, both obtained as macroscopic limits of the discrete BFtL multilane framework. Furthermore, both models can be interpreted as multilane extensions of the ARZ model, as illustrated in the right panel of Fig.~\ref{fig:relazioni}. The resulting macroscopic models are described by systems of balance laws, governing the evolution of vehicle density and average velocity for each lane. In Section \ref{sez:simulazioni}, we outline the numerical scheme employed for these models and present several numerical experiments to validate and explore their behavior. We also show comparisons with the microscopic multilane model and the first-order multilane model. Furthermore, we analyze scenarios both with and without lane-changing dynamics to highlight their impact on traffic flow. In this section, we also focus on constructing simulated fundamental diagrams based on our second-order macroscopic models. These diagrams are compared with empirical fundamental diagrams obtained from experimental traffic data. The comparison involves first analyzing and calibrating the models with respect to real-world measurements, then simulating traffic flow to generate fundamental diagrams, and finally evaluating the models' ability to reproduce key traffic features observed in practice, demonstrating how the structural assumptions of Model 2 ultimately yield a superior quantitative fit with empirical data. The paper concludes with Section \ref{sez:conclusions}, where we summarize the findings and discuss future research directions.

\section{Microscopic outlook}\label{sec:micro}
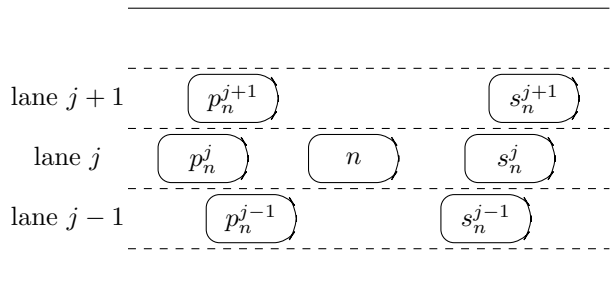
\begin{figure}[ht]
\centering
\begin{tikzpicture}[scale=0.8]
\draw [dashed](0,3)--(8,3);
\draw [dashed](0,2)--(8,2);
\draw [dashed](0,1)--(8,1);
\draw [dashed](0,0)--(8,0);
\draw (0,-1)--(8,-1);
\draw (0,4)--(8,4);

\path(-1,0.5)node{lane $j-1$};
\path (-1,1.5)node{lane $j$};
\path (-1,2.5)node{lane $j+1$};

 \draw[rounded corners=5pt] (1,2.1) -- (1+1.35,2.1) -- (2.5,2.1+0.23) -- (2.5,2.9-0.23) -- (1+1.35,2.9) -- (1,2.9) -- cycle;
\draw (1+1.5/2,2.5)node{$p_n^{j+1}$};

\draw[rounded corners=5pt]  (6,2.1) -- (6+1.35, 2.1) -- (7.5, 2.1+0.23) -- (7.5, 2.9-0.23) -- (6+1.35, 2.9) -- (6, 2.9) -- cycle;
\draw (6+1.5/2,2.5)node{$s_n^{j+1}$};

\draw[rounded corners=5pt] (.5,1.1) -- (.5+1.35,1.1) -- (2,1.1+0.23) -- (2,1.9-0.23) -- (.5+1.35,1.9) -- (.5,1.9) -- cycle;
\draw (0.5+1.5/2,1.5)node{$p_n^{j}$};

\draw[rounded corners=5pt] (3,1.1) -- (3+1.35,1.1) -- (4.5,1.1+0.23) -- (4.5,1.9-0.23) -- (3+1.35,1.9) -- (3,1.9) -- cycle;
\draw (3+1.5/2,1.5)node{$n$};

\draw[rounded corners=5pt] (5.6,1.1) -- (5.6+1.35,1.1) -- (5.6+1.5,1.1+0.23) -- (5.6+1.5,1.9-0.23) -- (5.6+1.35,1.9) -- (5.6,1.9) -- cycle;
\draw (5.6+1.5/2,1.5)node{$s_n^{j}$};

\draw[rounded corners=5pt] (1.3,0.1) -- (1.3+1.35,0.1) -- (2.8,0.1+0.23) -- (2.8,0.9-0.23) -- (1.3+1.35,.9) -- (1.3,.9) -- cycle;
\draw (1.3+1.5/2,0.5)node{$p_n^{j-1}$};

\draw[rounded corners=5pt] (5.2,0.1) -- (5.2+1.35,0.1) -- (5.2+1.5,0.1+0.23) -- (5.2+1.5,.9-0.23) -- (5.2+1.35,.9) -- (5.2,.9) -- cycle;
\draw (5.2+1.5/2,0.5)node{$s_n^{j-1}$};

\end{tikzpicture}
\caption{Schematic representation of a multilane road. Here, the reference vehicle is $n$, traveling in lane $j$, while $p_n^{k}$ and $s_n^{k}$ represent the vehicles just behind and in front of vehicle $n$ in lane $k=j-1,j,j+1$, respectively.}\label{strada}
\end{figure}

The starting point for the derivation of the second-order macroscopic models presented in this work is the multilane BFtL microscopic hybrid model introduced in \cite{piu1}, which serves as the reference framework for the upscaling procedure.

We consider a single population of \(N\) vehicles of equal length \(l > 0\) moving on a road with \(J\) lanes, each characterized by a distinct desired velocity profile \(\tilde{V}_j(\cdot)\). For each vehicle \(n = 1, \dots, N\), we label the adjacent vehicles as illustrated in Fig.~\ref{strada}. Moreover, \(I_j(t)\) denotes the set of indices of vehicles in lane \(j\), ordered by position, while \(N_j(t)\) represents the number of vehicles in lane \(j\) at time \(t\).

In the same spirit as in \cite{piu1}, we consider the following second-order dynamics: the multilane model consists of \(2N\) ODEs arranged in \(J\) subsystems, where the dynamics for each vehicle \(n \in I_j\) are given by
\begin{equation}
\label{bftl_nlane}
\begin{cases}
\dot{x}_n = v_n \\
\dot{v}_n = a_j(x_n,v_n) \\
\text{Change lane to $j'$ if:} \\
\quad a_{j'}(x_n,v_n) > (1+\eta) a_j(x_n,v_n) \quad \text{(incentive criterion)}, \\
\quad \Delta x^{j'}_n > l + d_s \ \text{and} \ x_n - x_{p_n^{j'}} > l + d_s \quad \text{(safety criterion)}.
\end{cases}
\end{equation}with acceleration  $$a_j(x_n,v_n):= \alpha \left(\tilde{V}_j(\Delta x_n^j) - v_n\right) + \beta \dfrac{\Delta v_n^j}{(\Delta x_n^j-K)^{\gamma+1}} $$ and  \(\Delta x_n^j = x_{s_n^j} - x_n\), \(\Delta v_n^j = v_{s_n^j} - v_n\). Differing from the original model in \cite{piu1}, we have introduced the constant $K \in \{0,\, l + d_s\}$, which, similarly to what happens in \eqref{KBFTL}, allows us to consider both the ``classical'' and the ``modified'' version of the acceleration term.

Lane changes are assumed to be instantaneous and occur whenever both the incentive criterion, which requires a velocity gain with a factor \(\eta \geq 0\), and the safety criterion, which ensures a minimum safe distance \(d_s > 0\), are satisfied. Note that the factor \(\eta\) is newly introduced here, while the original model in \cite{piu1} corresponds to the case \(\eta = 0\). Moreover, the model incorporates a stochastic component to reflect the infrequent nature of lane changes, as highlighted by experimental studies \cite{studi1,studi2,hertyvisconti}. Therefore, the lane-changing dynamics are governed by two stochastic processes: an expected number of lane changes \(N_{lc}\) per unit of time is fixed, and \(N_{lc}\) candidate vehicles are randomly selected. The timing of each lane change is then randomly assigned. The first process can be described as a Bernoulli random variable over the set of vehicles, while the second is approximated by uniformly distributing the \(N_{lc}\) events over the unit time interval. As a result, a lane change is expected approximately every \(\tau_{lc} = \frac{1}{N_{lc}}\) seconds.

For all further details on the coupling between the hybrid system and the described stochastic processes, we refer the reader to the already cited work~\cite{piu1}.

%
\section{Two novel second-order macroscopic multilane models}\label{sez:macromod}

In this section, we present the derivation of two second-order macroscopic multilane models inspired by the microscopic dynamics described by the hybrid microscopic model \eqref{bftl_nlane}. Our aim is to ``translate'' the lane-changing processes we described in the microscopic model into macroscopic terms, providing the evolution laws of traffic through the usual aggregate variables of density and average speed. For this purpose, it is useful to provide a definition of a discrete density.
\begin{definition} \label{def:discrete:density} The local density in front of the $n$-th vehicle in lane $j$ at time $t$ is given by:
\begin{equation*}
	\rho^{(n)}_j(t):=\frac{l+d_s}{\Delta x^j_n}.
\end{equation*}
\end{definition}
Through this definition, we can now introduce the macroscopic variables of density and velocity for each lane as piecewise constant reconstructions of the local densities and velocities.

\begin{definition}\label{piece} The macroscopic densities and the mean velocities are given by:
\begin{equation*}
\rho_{j}(x,t):=\rho_{j}^{(n)}(t), \ v_j(x,t):=v_n(t), \quad n \in I_j, \, x \in [x_n,x_{s^{j}_n}), \, j=1,\dots,J,
\end{equation*}
where $I_j=I_j(t)$ is the label set of the vehicles in lane $j$. 
\end{definition}

Following the methodology outlined in \cite{piu2}, which developed an equation for the evolution of the density $\rho_j$ with lane changing, we intend to derive here an equation for the velocity $v_j$  using two different approaches, resulting in two models referred to as Model 1 and Model 2.

To this end, we distinguish between two different scenarios as consequence of a lane change, see also Fig.~\ref{scenario}:
\begin{description}
\item[Gain scenario.] A vehicle, which we will refer to as $n$, enters lane $j$ in front of vehicle $p^{j}_n$. See the left panel of Fig.~\ref{scenario}. This results in a gain for lane $j$ and in a loss for lane $j'$.
\item[Loss scenario.] The vehicle in front of vehicle $p^{j}_n$, denoted as $n$, leaves lane $j$. See the right panel of Fig.~\ref{scenario}. This results in a gain for lane $j'$ and in a loss for lane $j$.
\end{description}

In both models we will derive, the mass conservation equation includes a source term $S_j$, see eq.~\eqref{defS}, identical to that of the model in \eqref{macro_1_ord}, but with the key difference that velocity is now treated as a separate variable, independent of density.

Since the density equation is derived from mass conservation, the total mass remains a conserved quantity. Therefore, the gain and loss terms for the density equation are perfectly symmetric. In contrast, velocity is not conserved, requiring separate treatment for the gain and loss scenarios in the velocity equation.

It is worth pointing out that the derivation of a macroscopic system from microscopic dynamics is generally not uniquely defined, as it depends on the chosen closure variables. In the following, we will derive two distinct macroscopic models which differ precisely in this aspect: the first is closed by deriving a balance equation directly for the velocity (requiring specific assumptions on the post-transition speed), while the second is closed by tracking the evolution of the generalized momentum $y$, allowing the velocity to be recovered indirectly without further assumptions.

\begin{figure}[t]
\centering
\begin{tikzpicture}[scale=0.8]
\draw [dashed](0,3)--(6,3);
\draw[dashed](0,2)--(6,2);
\draw [dashed](0,1)--(6,1);
\draw [dashed](7,3)--(6+7,3);
\draw [dashed](7,2)--(6+7,2);
\draw [dashed](7,1)--(6+7,1);
\path (-1,1.5)node{lane $j'$};
\path (-1,2.5)node{lane $j$};
\draw[rounded corners=5pt] (0.5,2.1) -- (0.5+.9,2.1) -- (1.5,2.1+0.23) -- (1.5,2.9-0.23) -- (0.5+.9,2.9) -- (0.5,2.9) -- cycle;
\draw (0.5+1/2,2.5)node{$p_n^{j}$};
\draw[rounded corners=5pt] (4,2.1) -- (4+.9,2.1) -- (5,2.1+0.23) -- (5,2.9-0.23) -- (4+.9,2.9) -- (4,2.9) -- cycle;
\draw (4+1/2,2.5)node{$s_n^{j}$};
\draw[rounded corners=5pt] (2,1.1) -- (2+.9,1.1) -- (3,1.1+0.23) -- (3,1.9-0.23) -- (2+.9,1.9) -- (2,1.9) -- cycle;
\draw (2+1/2,1.5)node{$n$};
\draw[->,ultra thick](2+1/2,1.9)--(2+1/2+0.2,2.6);
\draw[rounded corners=5pt] (0.5+7,2.1) -- (0.5+7+.9,2.1) -- (1.5+7,2.1+0.23) -- (1.5+7,2.9-0.23) -- (0.5+7+.9,2.9) -- (0.5+7,2.9) -- cycle;
\draw (7+0.5+1/2,2.5)node{$p_n^{j}$};
\draw[rounded corners=5pt] (7+4,2.1) -- (7+4+.9,2.1) -- (7+5,2.1+0.23) -- (7+5,2.9-0.23) -- (7+4+.9,2.9) -- (7+4,2.9) -- cycle;
\draw (7+4+1/2,2.5)node{$s_n^{j}$};
\draw[rounded corners=5pt] (7+2,2.1) -- (7+2+.9,2.1) -- (7+3,2.1+0.23) -- (7+3,2.9-0.23) -- (7+2+.9,2.9) -- (7+2,2.9) -- cycle;
\draw (7+2+1/2,2.5)node{$n$};
\draw[->,ultra thick](7+2+1/2,2.1)--(7+2+1/2+0.2,2.1-0.7);

\node at (3, 3.4) {``GAIN''};
\node at (10, 3.4) {``LOSS''};
\end{tikzpicture}
\caption{Left: ``Gain scenario'', the vehicle $n$ enters lane $j$. Right: ``Loss scenario'', the vehicle $n$ leaves lane $j$. Note that since $x_n$ denotes the position of the rear of vehicle $n$, the space in front of it must be at least $l+d_s$.}\label{scenario}
\end{figure}
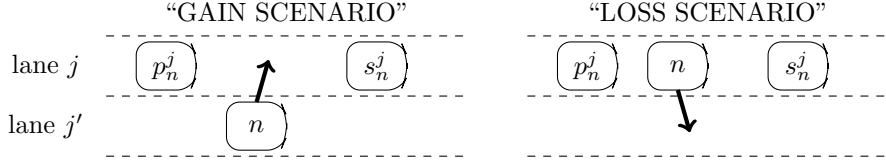

\subsection{Derivation of Model 1: evolution of the variable $v_j$ with lane changes}

The approach used to derive Model 1 involves analyzing the evolution of the velocity in the presence of a lane change, employing directly the motion equations provided by the microscopic model \eqref{bftl_nlane}.

\paragraph*{A discrete description}

Let us assume that at time $t$, a lane change occurs in front of a vehicle indexed as $p^{j}_n$, which is traveling in lane $j$, as depicted in Fig.~\ref{scenario}.

From the perspective of vehicle $p^{j}_n$, at time \( t^- \) (i.e., before the lane change occurs), its acceleration \( \ddot{x}_{p^{j}_n}(t^-) \) is determined by the interaction with the vehicle in front of it. After the lane change, at time \( t^+ \), this interaction changes because its leading vehicle has changed. In particular, from equations \eqref{bftl_nlane}, we have:
\begin{align*}
\ddot{x}_{p^{j}_n}(t^-) &= \alpha \left( \tilde{V}_{j}\left(x_{\star}(t^-)-x_{p^{j}_n}(t^-)\right)-v_{p^{j}_n}(t^-)\right) \\
&\quad + \beta \left( \dfrac{v_{\star}(t^-)-v_{p^{j}_n}(t^-)}{(x_{\star}(t^-)-x_{p^{j}_n}(t^-)-K)^{\gamma+1}} \right), \\
\ddot{x}_{p^{j}_n}(t^+) &= \alpha \left( \tilde{V}_{j}\left(x_{\star\star}(t^+)-x_{p^{j}_n}(t^+)\right)-v_{p^{j}_n}(t^+)\right) \\
&\quad + \beta \left( \dfrac{v_{\star\star}(t^+)-v_{p^{j}_n}(t^+)}{(x_{\star\star}(t^+)-x_{p^{j}_n}(t^+)-K)^{\gamma+1}} \right).
\end{align*}
where in the Gain case we have $\star=s^{j}_n$ and $\star \star=n$, whereas in the Loss case we have $\star=n$ and $\star \star=s^{j}_n$.

In order to compute the difference \( \Delta \ddot{x}_{p^{j}_n} := \ddot{x}_{p^{j}_n}(t^+) - \ddot{x}_{p^{j}_n}(t^-) \), we formulate some assumptions that allow us to express all variables at time \( t^+ \) in terms of the variables at time \( t^- \).

\begin{description}
\item[(V1)] Since \( p^{j}_n \) does not change lane at time \( t \), its position is certainly a continuous function, thus \( x_{p^{j}_n}(t^+) = x_{p^{j}_n}(t^-)=x_{p^{j}_n}(t) \). The same holds for \( s^{j}_n \), which cannot be influenced by the movement of $n$.

\item[(V2)] When the lane change takes place, the velocity of vehicle \( p^{j}_n \) does not change instantaneously, so \( v_{p^{j}_n}(t^+) = v_{p^{j}_n}(t^-)=v_{p^{j}_n}(t) \). The same holds for \( s^{j}_n \).

\item[(V3)] At the moment the lane change occurs, the density in front of the vehicle \( p_n^j \) changes instantaneously. Consequently, we assume that the primary factor contributing to the change in acceleration is the relaxation term. At this stage, we neglect the interaction with the vehicle \( \star \star \), setting \( \beta = 0 \) at time \( t^+ \). Indeed,  while the increase in density of the gain term requires an immediate reaction (braking), the sudden decrease in density brought by the loss term does not imply an instantaneous acceleration. In fact, an immediate braking is required by safety.

\end{description}

The quantity \( \Delta \ddot{x}_{p^{j}_n} \) can therefore be expressed as follows:
\begin{equation}\label{differenza_acc}
\begin{split}
\Delta \ddot{x}_{p^{j}_n} &= \alpha \Bigl( V_{j}\bigl( \rho_{j}^{(p^{j}_n)}(t^+) \bigr) - V_{j}\bigl( \rho_{j}^{(p^{j}_n)}(t^-) \bigr) \Bigr) \\
&\quad -\beta \frac{v_{\star}(t^-) - v_{p^{j}_n}(t^-)}{(l+d_s)^{\gamma+1}} \left( \frac{1}{\rho_{j}^{(p^{j}_n)}(t^-)} - \frac{K}{l+d_s} \right)^{-(\gamma+1)}.
\end{split}
\end{equation}
Notice that in the previous expression we have replaced $\tilde{V}_j(\cdot)$ with $V_j(\cdot)$ which is a function of the density such that
$$
    V_j(\rho_j^{(n)}(t)) = \tilde{V}_j(\Delta x_n^j(t)), \quad \forall n\in I_j, \ j\in\{1,\dots,J\}, \ t\in\mathbb{R}^+.
$$
This choice is made with the aim of computing the macroscopic limit, where the main variable is the density. In equation~\eqref{differenza_acc} we can specify the density at $t^-$ using Definition~\ref{def:discrete:density}:
$$
\rho_{j}^{(p^{j}_n)}(t^-) = \frac{l+d_s}{x_\star(t^-)-x_{p^{j}_n}(t^-)}.
$$
Instead, to specify the density at time $t^+$ we make an additional assumption for the Gain case following~\cite{piu2}:
\begin{description}
\item[(V4)] Since the position of \( n \) in the new lane can fall within the interval \( \left[x_{p^{j}_n}(t^-) + l + d_s, \, x_{s^{j}_n}(t^-) - l - d_s \right] \), we assume that ${x}_n(t^+) $ is a convex combination of the positions of vehicles \( p_n^j \) and \( s_n^j \). Therefore, we write:  
$ {x}_n(t^+) := \lambda \big(x_{s^{j}_n}(t^-) - l - d_s\big) + (1-\lambda) \big(x_{p^{j}_n}(t^-) + l + d_s\big)$, where the parameter $ \lambda \in [0,1]$ is assumed to depend on the density in the incoming lane of vehicle \( n \), specifically \( \lambda = \lambda(\rho_{j'}^{(n)}(t^-)) \).
\end{description} With this choice, we can easily express the densities at time $t^+$. In fact, as shown in \cite{piu2}, the density \( \rho_{j}^{(p^{j}_n)}(t^+) \) can be computed as
\begin{equation}\label{nuovadens}
\rho_{j}^{(p^{j}_n)}(t^+) =
\begin{cases} 
\rho_{j}^{(p^{j}_n)}(t^-) + G\left( \rho_{j'}^{(n)}(t^-), \rho_{j}^{(p^{j}_n)}(t^-) \right) \rho_{j}^{(p^{j}_n)}(t^-) & \text{(Gain case)}, \\[10pt]
\rho_{j}^{(p^{j}_n)}(t^-) - G\left( \rho_{j}^{(n)}(t^-), \rho_{j'}^{(n)}(t^-) \right) \rho_{j'}^{(n)}(t^-) & \text{(Loss case)},
\end{cases}
\end{equation}
Note that, to avoid any ambiguity with the notation for the matrix $A(U_j)$ introduced later in Section 3.4, we denote this amplification factor by $G$, which corresponds to the function denoted by $A$ in \cite{piu2}. For a possible choice for $G$ see eq.~\eqref{sceltaA}.

\paragraph*{A continuous description}\label{continuum}
Using the macroscopic variables \(\rho_j\) and \(v_j\) defined in Definition \ref{piece}, equation~\eqref{differenza_acc} can be reformulated as follows. Let us focus, for the moment, on the Gain case, i.e., \(\star = s^{j}_n\). The microscopic acceleration \(\ddot{x}_{p^{j}_n}(t)\) can be expressed as \(\frac{d}{dt}v_j(x_{p^{j}_n}(t),t) = \partial_t v_j + v_j \partial_x v_j\), where \(v_j \equiv v_j(x_{p^{j}_n}(t), t) = v_{p^{j}_n}(t)\). The velocity difference appearing in the term multiplied by \(\beta\) in \eqref{differenza_acc} can be rewritten as: \[
\begin{aligned}
    v_{s^{j}_n}(t) - v_{p^{j}_n}(t) &= \frac{\textnormal{d}}{\textnormal{d}t}\left(x_{s^{j}_n}(t) - x_{p^{j}_n}(t)\right) 
    = \frac{\textnormal{d}}{\textnormal{d}t}\left(\frac{l + d_s}{\rho_{j}^{(p^{j}_n)}(t)}\right) \\
    &\stackrel{\text{(A)}}{=} -\frac{l + d_s}{\rho_j^2}\left(\partial_t \rho_j + v_j \partial_x \rho_j \right) 
    \stackrel{\text{(B)}}{=} \frac{l + d_s}{\rho_j} \partial_x v_j,
\end{aligned}
\] where in step (A) we have set \(\rho_j \equiv \rho_j(x_{p^{j}_n}(t), t) = \rho_{j}^{(p^{j}_n)}(t)\), while step (B) follows from the mass conservation equation. Additionally, we denote the macroscopic density after the lane change as \(\rho_j^G\). Consequently, the right-hand-side of equation~\eqref{differenza_acc} can be rewritten in terms of macroscopic variables as:
\begin{equation}\label{intermedio}
\alpha\left(V_{j}\left(\rho_j^G\right) - V_{j}\left(\rho_{j}\right)\right) - \frac{\beta}{(l + d_s)^\gamma}\frac{1}{\rho_j}\frac{1}{\left(\frac{1}{\rho_j} - \frac{K}{l + d_s}\right)^{\gamma+1}} \partial_x v_j.
\end{equation}

Now, defining the pressure function \(P(\rho_j)\) as:\begin{equation}\label{defP}
P(\rho_j) = \begin{cases}
\frac{\beta}{\gamma (l + d_s)^\gamma} \rho_j^\gamma, & \text{if } K = 0,\\
\frac{\beta}{\gamma (l + d_s)^\gamma} \left(\frac{1}{\rho_j} - \frac{1}{\rho^{\text{max}}}\right)^{-\gamma}, & \text{if } K = l + d_s,
\end{cases}
\end{equation} equation \eqref{intermedio} can be further rewritten in terms of the derivative of this pressure function:
\begin{equation}\label{sorgenti_modello1_G}
\mathcal{G}^v(\rho_j, \rho_{j'}, \partial_x v_j): = \alpha\left(V_{j}\left(\rho_j^G\right) - V_{j}\left(\rho_{j}\right)\right) - \rho_j P'(\rho_j) \partial_x v_j.
\end{equation}  This expression provides insight into the acceleration variation in the gain case for lane \(j\).

Repeating the same analysis for the loss case, from equation~\eqref{differenza_acc} we can define
\begin{equation}\label{sorgenti_modello1_B}
\mathcal{L}^v(\rho_j,\rho_{j'},\de_x v_j) := \alpha\left(V_{j}\left(\rho_j^L\right)-V_{j}\left(\rho_{j}\right)\right) -\rho_jP'(\rho_j)\de_x v_j.
\end{equation}
where $\rho_j^L$ (and $\rho_j^G$) represents the macroscopic counterparts of \eqref{nuovadens}:\begin{equation*}
\begin{array}{l}
\rho_j^G=\rho_j+G(\rho_{j'},\rho_j)\rho_j,\\
\rho_j^L=\rho_j-G(\rho_{j},\rho_{j'})\rho_{j'}.
\end{array}
\end{equation*}

We now define two random variables $X\sim\text{Ber}(\pi^{j' \to j})$ and $Y\sim\text{Ber}(\pi^{j \to j'})$  with Bernoulli distribution, where $\pi^{h \to k}$ is the probability of a lane change from lane $h$ to lane $k$ within a unit of time. Hence, $X$ will represent lane changes directed towards lane $j$, while $Y$ will denote lane changes originating from lane $j$.

Using a Taylor expansion for the time derivative with $\Delta t>0$ in the second equation of \eqref{eq:arz} for a multilane flow, and considering the stochastic events due to the lane changes, the evolution of the velocity at time $t+\Delta t$ is provided by the stochastic variable $\tilde{v}_j(x, t+\Delta t; X, Y)$, which accounts for the possible lane-changing events through the Bernoulli random variables $X$ and $Y$: \begin{equation*}\label{casi_vel}
\begin{split}
\tilde{v}_j(x,t+\Delta t;X,Y) = v_j &+ \Delta t \big( - \left(v_j-\rho_j P'(\rho_j) \right) \partial_x v_j \\
&\quad + \alpha(V_j(\rho_j)-v_j) + X \cdot \mathcal{G}^v + Y\cdot \mathcal{L}^v \big) 
+ o(\Delta t),
\end{split}
\end{equation*} where $o(\Delta t)$ denotes the standard Landau notation for higher-order terms such that $\lim_{\Delta t \to 0} \frac{o(\Delta t)}{\Delta t} = 0$. Computing the  expected value of $\tilde{v}_j$  we obtain
\begin{equation*} 
\begin{split}
v_j (x,t+\Delta t)
&= \mathbb{E}[\tilde{v}_j(x,t+\Delta t;X,Y)]  \\ 
&= v_j + \Delta t \big( -\left(v_j-\rho_j P'(\rho_j) \right) \partial_x v_j 
\\
&\quad + \alpha(V_j(\rho_j)-v_j) 
+ \pi^{j' \to j} \mathcal{G}^v 
+ \pi^{j \to j'} \mathcal{L}^v \big) 
+ o(\Delta t),
\end{split}
\end{equation*} and finally, passing to the limit for $\Delta t \to 0$ we get the equation
\begin{equation*}
\de_t v_{j} +  \left(v_j-\rho_j P'(\rho_j) \right) \de_x v_j=\alpha(V_j(\rho_j)-v_j)+ \pi^{j' \to j} \mathcal{G}^v + \pi^{j \to j'} \mathcal{L}^v,
\end{equation*} which describes the average evolution of velocity in the presence of lane changes, accounting for both gain and loss events. Explicitly expressing the terms $\mathcal{G}^v$ and $\mathcal{L}^v$ as in Eqs. \eqref{sorgenti_modello1_G} and \eqref{sorgenti_modello1_B}, the previous equation becomes: \begin{equation*}
\partial_t v_{j} +  \left(v_j - \rho_j P'(\rho_j)(1-C_j) \right) \partial_x v_j = \alpha(V_j(\rho_j) - v_j + R_j),
\end{equation*} 
where $C_j=C_j(\rho_{j-1},\rho_{j},\rho_{j+1},v_{j-1},v_j,v_{j+1})$, $R_j=R_j(\rho_{j-1},\rho_{j},\rho_{j+1},v_{j-1},v_j,v_{j+1})$, and
\begin{equation} \label{defCR}
\begin{aligned}
C_j &= \sum\limits_{j' \in T_j}\left(\pi^{j' \to j} + \pi^{j \to j'}\right), \\
R_j &= \sum\limits_{j' \in T_j}\left(\pi^{j' \to j}\left(V_{j}\left(\rho_j^G\right) - V_{j}\left(\rho_{j}\right)\right) + \pi^{j \to j'}\left(V_{j}\left(\rho_j^L\right) - V_{j}\left(\rho_{j}\right)\right)\right), 
\end{aligned}
\end{equation} 
for $(x,t) \in \mathbb{R} \times \mathbb{R}^+$ and for all $j = 1, \dots, J$.

With the notation $T_j$ we denote the set of the interacting lanes:  $T_j = \{j-1, j+1\} \text{ for } 1 < j < J, \quad T_1 = \{2\}, \quad T_J = \{J-1\}$.

This equation, together with equation \eqref{macro_1_ord} and with $f_j(\rho_j,v_j)=\rho_j v_j$, forms the second-order macroscopic model, which we refer to as the \textit{Model 1}. It should be emphasized that, in the source term of the density equation, velocity is now an independent variable, distinct from density, and evolves according to its own evolution equation derived in this section. Consequently, the source term $S_j$ takes the form $S_j=S_j(\rho_{j-1},\rho_{j},\rho_{j+1},v_{j-1},v_j,v_{j+1})$.

\subsection{Derivation of Model 2: evolution of the variable $y_j$ with lane changes}
We now propose a different approach to derive a second-order macroscopic multilane model, starting from the underlying microscopic dynamics. It is known that the macroscopic limit of the system \eqref{KBFTL} corresponds to the classical ARZ model \eqref{eq:arz}, which can be rewritten in conservative form as in \eqref{eq:arz_cons}. The core idea is to study the evolution of \( y \) in scenarios involving lane changes, and to derive a corresponding macroscopic equation that can be coupled with the density evolution equation, in a manner similar to the approach used for Model 1.

To begin this analysis, we first consider the discretized representation of the variable \( y \).

\begin{definition} 
The local value \( y_j^{(n)} \) associated with vehicle \( n \) in lane \( j \) at time \( t \) is defined as:
\begin{equation*}
y_j^{(n)} := \rho_j^{(n)}(t) \left( v_n(t) + P\left(\rho_j^{(n)}(t)\right) \right).
\end{equation*}
\end{definition}

Given that the system typically includes a relaxation term (\(\alpha \neq 0\)), \( y \) is not a conserved variable. Therefore, we analyze the gain and loss terms separately to better understand their effects on the evolution of \( y \).

\paragraph*{A discrete description}
Let us consider the two scenarios depicted in Fig.~\ref{scenario}. In both cases, the difference $\Delta y_j^{p^j_n} := y_j^{(p^{j}_n)}(t^+) - y_j^{(p^{j}_n)}(t^-)$ is defined as:
\begin{equation}\label{ydiff}
\Delta y_j^{p^j_n} = \rho_j^{(p^{j}_n)}(t^+) \left( v_{p^{j}_n}(t^+)+P\left(\rho_j^{(p^{j}_n)}(t^+)\right)\right)-y_j^{(p^{j'}_n)}(t^-)
\end{equation} where at time $t^+$ the density is computed as in equation \eqref{nuovadens}. Similarly to what we have assumed for Model 1, we consider the assumptions (V1), (V2), and (V4) to be valid. Note that, in this case, by working directly with \( y_j \), we do not need the assumption (V3) that is required for the analysis conducted in the study of the evolution of the velocity. In fact, (V3) allowed us to avoid assumptions on the velocity $v_{\star\star}$ at time $t^+$.

\paragraph*{A continuous description}
In a manner similar to the derivation of Model 1, we can express the right-hand side of equation \eqref{ydiff} in terms of macroscopic variables $\rho_j, v_j, y_j$, yielding two terms that describe the variation of $y$ in the presence of lane-changing events, specifically in the gain and loss cases, respectively:
\begin{align*}
\mathcal{G}^y(\rho_j,\rho_{j'},v_j) &:=\rho_{j}^G \left(v_j+P\left(\rho_j^G\right)\right)-y_j\\
\mathcal{L}^y(\rho_j,\rho_{j'},v_j) &:= \rho_{j}^L \left(v_j+P\left(\rho_j^L\right)\right)-y_j
\end{align*}

We consider the two Bernoullian random variables $X \sim \text{Ber}(\pi^{j' \to j})$ and $Y \sim \text{Ber}(\pi^{j \to j'})$ as before, and we perform a Taylor expansion for the time derivative with $\Delta t > 0$ in the second equation of \eqref{eq:arz_cons}. This results in a stochastic evolution given by
\begin{equation*}\label{casi_vel2}
\begin{split}
\tilde{y}_j(x,t+\Delta t;X,Y) 
&= y_j + \Delta t \big( -\partial_x (y_j v_j) + \alpha \rho_j (V(\rho_j) - v_j) \\
&\quad + X \cdot \mathcal{G}^y + Y \cdot \mathcal{L}^y \big) + o(\Delta t),
\end{split}
\end{equation*} By computing the expected value of $\tilde{y}_j$ and taking the limit as $\Delta t \to 0$, we obtain:
\begin{equation*}
\partial_t y_j + \partial_x (y_j v_j) = \alpha \rho_j (V(\rho_j) - v_j) + \pi^{j' \to j} \mathcal{G}^y + \pi^{j \to j'} \mathcal{L}^y.
\end{equation*} Furthermore, recalling the definition of $y_j$, and incorporating the explicit expressions for $\mathcal{G}^y$ and $\mathcal{L}^y$, the equation can be rewritten as:

\begin{equation}\label{finalemod2}
\partial_t y_{j} +  \partial_x \left(\frac{y_j^2}{\rho_j}-y_j P(\rho_j)  \right) = \alpha \rho_j \left(V_j(\rho_j) - \frac{y_j}{\rho_j}+P(\rho_j)\right) + Q_j,
\end{equation} where $Q_j=Q_j(\rho_{j-1},\rho_{j},\rho_{j+1},v_{j-1},v_j,v_{j+1})$ with \begin{equation}\label{defQ}
Q_j = \sum\limits_{j' \in T_j} \left( \pi^{j' \to j} \left( \rho_{j}^G \left( v_j + P(\rho_j^G) \right) - y_j \right) + \pi^{j \to j'} \left( \rho_{j}^L \left( v_j + P(\rho_j^L) \right) - y_j \right) \right),
\end{equation}
for $(x,t) \in \mathbb{R} \times \mathbb{R}^+$, and for all $j = 1, \dots, J$. This equation, in conjunction with equations \eqref{macro_1_ord} and \eqref{defS} and with the flux function $f_j(\rho_j, y_j) = y_j - \rho_j P(\rho_j)$, forms the second-order macroscopic model, referred to as \textit{Model 2}.

It is noteworthy that the hyperbolic part of this model is already written in a conservative form. However, it is also possible to rewrite the equations in terms of the variables $\rho_j$ and $v_j$. In this case, equation \eqref{finalemod2} becomes:

\begin{equation*}
\partial_t v_{j} +  (v_j - \rho_j P'(\rho_j)) \partial_x v_j = \alpha (V_j(\rho_j) - v_j) + \frac{Q_j}{\rho_j} - \nu S_j \left( P'(\rho_j) + \frac{v_j + P(\rho_j)}{\rho_j} \right).
\end{equation*}

\subsection{Lane changing conditions} The probability of changing lanes is influenced by the conditions governing lane changes and can be modeled similarly to the approach in \cite{piu2}. In this context, we express these conditions in terms of macroscopic variables, namely density and velocity. Thus, the macroscopic counterpart of the rules governing the lane change from lane \(j\) to lane \(j'\), derived from the microscopic model in \eqref{bftl_nlane}, is as follows: \begin{equation*}\label{cambimicrose}
\begin{array}{ll}
v_{j'} > (1+\eta)v_{j} & \text{(incentive criterion)} \\
\rho_{j'} < \mu := \frac{1}{2} & \text{(safety criterion)} 
\end{array}
\end{equation*} where \(\mu\) represents a critical density beyond which lane changes become infeasible, because lane \(j'\) lacks sufficient space to accommodate additional vehicles. We observe that the incentive condition in the macroscopic formulation is based on velocities.

Therefore, the probability of jumping from lane $j$ to lane $j'$, $\tilde \pi^{j\to j'}$, is defined as
\begin{equation} \label{def:prob}
\begin{aligned}
\tilde \pi^{j \to j'}(\rho_{j'},\rho_{j}) &= g(\rho_{j'})\, \mathbf{1}_{lc}(\rho_{j'},\rho_j,v_j,v_{j'}), \\
\mathbf{1}_{lc}(\rho_{j'},\rho_j,v_j,v_{j'}) &= \max\left\{0, \min\left\{ \sign\left(v_{j'} - (1+\eta)v_{j}\right), \sign\left(\mu -\rho_{j'}\right) \right\}\right\}.
\end{aligned}
\end{equation} and $g$ is a smooth function that monotonically decreases with its argument, expressing the percentage of mass that actually changes lanes.

We observe that if for some \((x,t)\) the probability \(\tilde \pi^{j' \to j}\) is strictly positive, then from the incentive criterion it necessarily follows that the probability \(\tilde \pi^{j \to j'}\) is equal to zero.

In the case \(J > 2\), excluding lanes 1 and \(J\), four possible events can occur for a given state \((x, t)\): a lane change to either lane \(j+1\) or \(j-1\), a lane change from either lane \(j+1\) or \(j-1\), or the trivial event where no lane change occurs. Consequently, we assign four probabilities to each lane: two ``exit'' probabilities \(\tilde \pi^{j \to j'}\) and two ``entry'' probabilities \(\tilde \pi^{j' \to j}\), where \(j' \in T_j\). If more than one of these probabilities is strictly positive, the highest one is prioritized and the others are set to zero. Then, for $(h,k)\in\{(j,j+1),(j,j-1),(j-1,j),(j+1,j)\}$ we define
\begin{equation} \label{eq:probab}
    \pi^{h \to k}=\begin{cases}
    \tilde \pi^{h \to k} & \text{if } \tilde \pi^{h \to k}=\pi^{\text{max}}\\ 0 & \text{otherwise}\end{cases}
\end{equation}
where $\pi^{\text{max}}=\max\{\tilde \pi^{j \to j-1},\tilde \pi^{j \to j+1},\tilde \pi^{j-1 \to j},\tilde \pi^{j+1 \to j}\}$.

In case $\pi^{\text{max}}$ is not uniquely-defined, we will adopt the following rules.
\begin{itemize}
\item When a vehicle is equally likely to move left or right, it will be more likely to shift to the left: if \(\tilde \pi^{j \to {j-1}}=\tilde \pi^{j \to {j+1}}\), then \(\tilde \pi^{j \to {j-1}}=0\), meaning we prioritize the lane change to the left lane.
\item If incoming traffic from the left and right lanes has equal probability, preference is given to cars entering from the right: if \(\tilde \pi^{{j-1} \to j}=\tilde \pi^{{j+1} \to j}\), then \(\tilde \pi^{j \to {j+1}}=0\), meaning we prioritize lane changes originating from the right lane.
\item When the probability for an outgoing lane change matches the incoming flow, priority is given to outgoing changes, meaning vehicles will prefer to leave their current lane rather than allowing new traffic in: if \(\tilde \pi^{{j+1} \to {j}} = \tilde \pi^{j \to {j-1}}\), then \(\tilde \pi^{j \to {j+1}}=0\), and if \(\tilde \pi^{{j-1} \to {j}} = \tilde \pi^{j \to {j+1}}\), then \(\tilde \pi^{j \to {j-1}}=0\), meaning we prioritize the outgoing mass changes.
\end{itemize}

\subsection{The models}
The overall second-order models, derived in the previous sections, are given by the following $2J \times 2J$ systems. Employing the non-conservative variable $v_j$, we can write both Model 1 and Model 2 in quasi-linear form as
\begin{equation*}
	\partial_t U + \mathcal{A}(U) \partial_x U = \mathcal{B}(U),\, \text{ on } \ (x,t)\in\mathbb{R}\times \mathbb{R}^+.
\end{equation*}
where
\begin{equation*}
U=\begin{pmatrix}
U_1  \\ 
\vdots  \\ 
U_j \\
 \vdots  \\  
U_J\\ 
\end{pmatrix},\,
\mathcal{A}(U)=\begin{pmatrix}
A(U_1) & & & & 0 \\ 
& \ddots & & & \\ 
& & A(U_j) & & \\
& & & \ddots & \\  
0 & &  & & A(U_J)\\ 
\end{pmatrix},\,
\mathcal{B}(U)=\begin{pmatrix}
B(U_1)  \\ 
\vdots  \\ 
B(U_j) \\
 \vdots  \\  
B(U_J)\\ 
\end{pmatrix}
\end{equation*} with $U_j=(\rho_j,v_j)^T$. The two models are then specified as
	\begin{equation} \label{eq:model1:noncons}
	\begin{array}{ll}
		\hspace{-2cm} \text{\underline{Model 1}}  &  A(U_j)= \begin{pmatrix} 
			v_j & \rho_j \\ 
			0 & v_j-\rho_j P'(\rho_j)\left( 1- C_j\right) \\ 
		\end{pmatrix}\\[4ex]
		&
		B(U_j)= \begin{pmatrix} 
			\nu S_j \\ 
			\alpha(V_j(\rho_j)-v_j+R_j) \\ 
		\end{pmatrix}
	\end{array}
	\end{equation}
	\begin{equation} \label{eq:model2:noncons}
		\begin{array}{ll}
		\hspace{-0.6cm} \text{\underline{Model 2}}   & A(U_j)= \begin{pmatrix} 
			v_j & \rho_j \\ 
			0 & v_j-\rho_j P'(\rho_j) \\ 
		\end{pmatrix} \\[4ex] & B(U_j)= \begin{pmatrix} 
			\nu S_j \\ 
			\alpha (V_j(\rho_j) - v_j) + \frac{Q_j}{\rho_j}-\nu S_j \left(P'(\rho_j)+\frac{v_j+P(\rho_j)}{\rho_j}\right) \\ 
		\end{pmatrix}\\
	\end{array}
	\end{equation}

We note that in both models, the matrix $\mathcal{A}(U)$ is a block matrix on the main diagonal, where each $2 \times 2$ block is upper triangular.

To complete the definition of the models, we recall the definitions of the pressure $P$ in eq. \eqref{defP}, and the terms $C_j,R_j$ in eq. \eqref{defCR}, and $Q_j$ in eq. \eqref{defQ}. 

As in \cite{piu2}, the term $S_j$ is given by:
\begin{equation}\label{defS}
S_j  = \sum\limits_{j' \in T_j}\left(\pi^{j' \to j} G(\rho_{j'},\rho_{j})\rho_{j} -  {\pi^{j \to j'}}  G(\rho_{j},\rho_{j'})\rho_{j'}\right),
\end{equation} where the function \( G(\cdot,\cdot) \) denotes the amplification factor of the density in the target lane.  Specifically, it is defined as:\begin{equation} \label{sceltaA}
G(\rho_{h},\rho_{k}) = \left(\left(\lambda(\rho_{h}) + (1-2\lambda(\rho_{h}))\rho_{k}\right)^{-1}-1 \right),
\end{equation} 
with $\lambda(\rho)=1-\frac{\rho}{\rho^{\text{max}}}$.

The probability $\pi^{h \to k}$ is defined as in~\eqref{def:prob}-\eqref{eq:probab}, using the linear function $g(\rho)=1-2\frac{\rho}{\rho^{\text{max}}}$.

\paragraph*{Eigenstructure}
The eigenstructure of the two models can be easily recovered by the quasi-linear form and it is given by the $2J\times 2J$ matrices $\Lambda$ and $\Psi$ of the eigenvalues and associated eigenvectors, respectively:
\begin{equation*}
\Lambda=\begin{pmatrix}
\Lambda_1 & & & & 0 \\ 
& \ddots & & & \\ 
& & \Lambda_j & & \\
& & & \ddots & \\  
0 & &  & & \Lambda_J\\ 
\end{pmatrix}, \quad
\Psi = \begin{pmatrix}
\Psi_1 & & & & 0 \\ 
& \ddots & & & \\ 
& & \Psi_2 & & \\
& & & \ddots & \\  
0 & &  & & \Psi_J\\ 
\end{pmatrix},
\end{equation*}
where \begin{equation}\label{auto_tutto}
\begin{array}{lll}
\text{\underline{Model 1}}  &  \Lambda_j=\text{diag}\left(  v_j-\rho_jP'(\rho_j)(1-C_j), v_j \right),
 &  \Psi_j=\begin{pmatrix}
1 & 1 \\ -P'(\rho_j)(1-C_j) & 0 \end{pmatrix}, \\[5ex]
\text{\underline{Model 2}}   & \Lambda_j=\text{diag}\left(  v_j-\rho_jP'(\rho_j), v_j \right),    
& \Psi_j=\begin{pmatrix}
1 & 1 \\ -P'(\rho_j) & 0
\end{pmatrix}. \\
\end{array}
\end{equation}

\paragraph*{Hyperbolicity and anisotropy}
Note that the homogeneous part of Model 2 coincides with that of the ARZ model, thereby inheriting the properties of being a hyperbolic system for non-negative densities, and satisfying the anisotropy condition, namely $v_j-\rho_jP'(\rho_j)\leq v_j$. As for Model 1, we can demonstrate such properties as follows.
\begin{proposition}
	Model 1 is hyperbolic, and satisfies the anisotropic condition.
\end{proposition}
\begin{proof}
	Let us consider the matrix $A(U_j)$ in \eqref{eq:model1:noncons}. The eigenvalues of the matrix $A(U_j)$ are $\lambda_1= v_j-\rho_j P'(\rho_j)\left( 1- C_j\right)$ and $\lambda_2=v_j$. In particular $\lambda_1,\lambda_2\in\mathbb{R}$, hence the system is hyperbolic. Furthermore  the first characteristic velocity is slower than the velocity of the vehicles, given that $P'(r)>0$ for all $r\in[0,\rho^{\text{max}}]$, provided $\gamma>0$, and $C_j\in[0,1]$, hence $v_j-\rho_jP'(\rho_j)(1-C_j)\leq v_j $.
\end{proof}

We further note that $\lambda_1$ of Model~1 depends on $C_j$, which is determined by the lane-changing probabilities and hence by the source terms. Despite this unusual feature, hyperbolicity and the anisotropy condition remain valid for all $C_j \in [0,1]$.

Regarding the well-posedness of the systems \eqref{eq:model1:noncons} and \eqref{eq:model2:noncons}, we place the models within the classical framework of hyperbolic systems of balance laws \cite{dafermos2005hyberbolic,bressan2000hyperbolic}. As discussed, the convective part strictly preserves hyperbolicity provided that $\rho_j > 0$. Furthermore, we require the source terms $S_j, R_j$, and $Q_j$ to satisfy the property of local Lipschitz continuity with respect to the state variables. It is worth noting that while lane-changing events in the microscopic model are discrete and lead to discontinuous state transitions, the macroscopic limit effectively smooths these dynamics. This regularity is inherited from the smooth functions $V_j$ and $G$, and from the structure of the transition probabilities $\pi^{h \to k}$, which rely on the Lipschitz-continuous operator $\max(0, \cdot)$ to model speed gain incentives. Together, these properties fulfill the standard requirements for ensuring the local existence and uniqueness of solutions for sufficiently smooth initial data.

\paragraph*{Conservative Form}
Recall that, using equation \eqref{finalemod2}, Model 2 can be written in a conservative form through the variables \( W_j = (\rho_j, y_j) \), where \( y_j = \rho_j (v_j + P(\rho_j)) \). The equations for lane \( j = 1, \dots, J \)  can thus be formulated as the following system of balance laws:
\[
\partial_t W_j + \partial_x F(W_j) = B(W_j),
\]  
where the flux function \( F(W_j) \) and the source term \( B(W_j) \) are given by  
\begin{equation*} 
F(W_j) =
\begin{pmatrix} 
y_j - \rho_j P(\rho_j) \\ 
\frac{y_j^2}{\rho_j} - y_j P(\rho_j) 
\end{pmatrix}, \qquad
B(W_j) =
\begin{pmatrix} 
\nu S_j \\ 
\alpha \rho_j \left(V_j(\rho_j) - \frac{y_j}{\rho_j} + P(\rho_j) \right) + Q_j 
\end{pmatrix}.
\end{equation*}

Furthermore,  while the microscopic derivation implicitly requires strictly positive densities, the resulting macroscopic PDEs can be continuously extended to include the vacuum state ($\rho_j = 0$). This is achieved by conventionally setting the momentum $y_j = 0$ and the velocity $v_j = 0$ whenever $\rho_j = 0$, thus avoiding singularities. Moreover, since the source terms vanish as the density approaches zero, the system naturally prevents the onset of negative densities.

\paragraph*{Relaxation limit}
	We observe that in the relaxation limit $\alpha \to +\infty$, Model 2~\eqref{eq:model2:noncons} converges to the first-order model studied in~\cite{piu2}, i.e.
	$$
		\partial_t \rho_j + \partial_x ( \rho_j V_j(\rho_j) ) = \nu S_j.
	$$
	In contrast, the relaxation of Model 1~\eqref{eq:model1:noncons} leads to a first-order model with a modified velocity function. Indeed, one obtains
	$$
		\partial_t \rho_j + \partial_x ( \rho_j ( V_j(\rho_j) + R_j ) ) = \nu S_j,
	$$
	which means that the lane changes influence the equilibrium speed of the flow. In both multilane models, the variable \( v_j \) relaxes toward the equilibrium velocity \( V_j(\rho_j) \) over a timescale proportional to \( \alpha^{-1} \). Formally, when \(\alpha \to +\infty\), the relaxation term becomes dominant and imposes the constraint \(
v_j = V_j(\rho_j) + o(1),
\) which reduces the second-order model to a first-order conservation law. This type of singular limit is well-known in relaxation theory and has been extensively used in macroscopic traffic flow models, e.~g., \cite{arz1, colombo2002hyperbolic}, in order to investigate stability of equilibrium   solutions. In our case, the convergence can be understood as a formal limit under the assumption that the initial data are compatible with the equilibrium manifold \( v_j = V_j(\rho_j) \) (well prepared initial data), and that the solutions remain smooth in the limit.

\paragraph*{Equilibria} We briefly discuss the nature of steady states of the models \eqref{eq:model1:noncons} and \eqref{eq:model2:noncons}. We define a \emph{global equilibrium} or steady state of the multi-lane model as a state in which no lane changes occur and all lanes are in their \emph{local equilibrium}, i.e., the density and velocity are spatially homogeneous, see \cite{piu2}. More specifically, we provide the following definition.

\begin{definition} \label{def:ss}
A global equilibrium $\{(\overline{\rho}_j(x),\overline{v}_j(x))\}_{j=1}^J$ of the macroscopic multi-lane models \eqref{eq:model1:noncons} and \eqref{eq:model2:noncons} is characterized by
$$
  \overline{\rho}_j(x) \equiv \overline{\rho}_j\in (0,\rho^{\max}] \ \text{  and  } \overline{v}_j(x) \equiv \overline{v}_j\in (0,v_j^{\max}]
$$ 
and for Model 1 \eqref{eq:model1:noncons}
 \[
\begin{aligned}
  S_j(\overline{\rho}_{j-1}, \overline{\rho}_j, \overline{\rho}_{j+1}, \overline{v}_{j-1}, \overline{v}_j, \overline{v}_{j+1}) &= 0, \\
  C_j(\overline{\rho}_{j-1}, \overline{\rho}_j, \overline{\rho}_{j+1}, \overline{v}_{j-1}, \overline{v}_j, \overline{v}_{j+1}) &= 0, \\
  R_j(\overline{\rho}_{j-1}, \overline{\rho}_j, \overline{\rho}_{j+1}, \overline{v}_{j-1}, \overline{v}_j, \overline{v}_{j+1}) &= 0,
\end{aligned}
\]
 for all $j=1,\dots,J$, whereas for Model 2 \eqref{eq:model2:noncons}  \[
\begin{aligned}
  S_j(\overline{\rho}_{j-1}, \overline{\rho}_j, \overline{\rho}_{j+1}, \overline{v}_{j-1}, \overline{v}_j, \overline{v}_{j+1}) &= 0, \\
  Q_j(\overline{\rho}_{j-1}, \overline{\rho}_j, \overline{\rho}_{j+1}, \overline{v}_{j-1}, \overline{v}_j, \overline{v}_{j+1}) &= 0,
\end{aligned}
\]

for all $j=1,\dots,J$, i.e., in each lane, the density and velocity are constant, and no lane change occurs for these values.
\end{definition}

It is worth noting that, from a boundary-value problem perspective, such a homogeneous state can be globally observed on a finite road segment only under compatible boundary conditions (e.g., periodic boundaries for a ring road, or prescribed inflow values that exactly match the equilibrium states $\overline{\rho}_j$ and $\overline{v}_j$).

Let us verify that there exists at least one stationary state of the system that satisfies the equilibrium concept just introduced. Consider a traffic regime in which lane changes do not occur because none of the lane changing conditions are satisfied. In this case, the following conditions provide an example of equilibrium: \begin{equation*}
\begin{cases} 
\partial_x(\rho_j v_j) = 0 & \text{(C1)} \\ 
v_j = V_j(\rho_j) & \text{(C2)} \\ 
\partial_x v_j = 0  & \text{(C3)}
\end{cases} \Longrightarrow
\begin{cases} 
\partial_t \rho_j = 0 \\ 
\partial_t	 v_j = 0 
\end{cases} 
\end{equation*} We observe that condition (C3) imposes  a spatially homogeneous equilibrium velocity. Consequently, from condition (C1), this consideration extends to the density, yielding an equilibrium as defined in Definition \ref{def:ss}. Moreover, from condition (C2), we obtain that the velocity \(\overline{v}_j = V_j(\overline{\rho}_j)\). Therefore, the equilibria obtained in this manner, i.e., equilibria where the equilibrium velocity is given by the desired velocity profile \(V_j(\cdot)\), coincide with those derived for the first-order model \eqref{macro_1_ord}, as detailed in \cite{piu2}. The above conditions are sufficient but not necessary; indeed, other conditions could be found that satisfy the characterization given in Definition \ref{def:ss}.

The global equilibrium configurations described in Definition 4 are generally not unique. For a fixed total traffic mass, multiple distributions of vehicles across the lanes can result in a state where no lane-changing maneuvers are triggered (i.e., all probabilities $\pi^{h \to k}$ vanish). Any such configuration, provided it satisfies the local equilibrium condition $v_j = V_j(\rho_j)$, represents a valid stationary state of the system. Consequently, the final equilibrium reached depends on the initial state and the transient dynamics of the flow.

It is also important to recall that we consider a nonzero relaxation parameter, i.e., \(\alpha \neq 0\); otherwise, the system would admit additional equilibrium states, where the equilibrium velocities would not necessarily be determined by the velocity profile \(V_j(\cdot)\).

\paragraph*{Comparison and interpretation of the models}
It is instructive to highlight the fundamental differences between the two proposed systems. Model 1 acts as the direct projection of the microscopic second-order dynamics to the macroscopic scale. However, its mathematical structure is not completely satisfactory, as the lane-changing events intricately affect the transport term. Conversely, Model 2 builds upon the classical structure of the ARZ model, relying on the assumption that the property $y$ is conserved across lanes. This structural difference becomes particularly evident in their relaxation limits: while lane-changing in Model 1 directly alters the equilibrium velocity (leading to a non-standard first-order limit), Model 2 naturally relaxes to the first-order multilane model \eqref{macro_1_ord}. Ultimately, both models exhibit very similar qualitative behavior in capturing traffic phenomena, our numerical experiments in Sec.~\ref{sez:simulazioni} demonstrate that the heuristic assumption underlying Model 2 ultimately yields a better quantitative fit with real-world empirical data.

\section{Numerical aspects and experiments}\label{sez:simulazioni}
In this section, we present the numerical scheme used to integrate the multilane macroscopic models, with particular focus on the non-conservative formulations given by equations \eqref{eq:model1:noncons} and \eqref{eq:model2:noncons}.

We then illustrate a series of test cases to demonstrate the application of the proposed second-order macroscopic models in realistic traffic scenarios. These tests include comparisons with the microscopic multilane model, the first-order multilane model, and the second-order single-lane model without lane changes.

Except for Test 5, the desired velocity profiles are assumed to have a linear dependence on density:
\begin{equation}
\label{lineari}
V_j(\rho_j) = v_j^{\mathrm{max}} \left(1 - \frac{\rho_j}{\rho^{\mathrm{max}}} \right),
\end{equation}
for \( j = 1, \dots, J \), with \( 0 \leq v_1^{\mathrm{max}} \leq \dots \leq v_J^{\mathrm{max}} \).  
With these definitions, the resulting theoretical fundamental diagrams \( f_j(\rho_j) = \rho_j V_j(\rho_j) \) take the parabolic (Greenshields) shape. Furthermore we consider $\rho^{\mathrm{max}}=1$. 

Specific details for each configuration are provided in the descriptions of the individual tests.

\subsection{Numerical scheme}\label{subsec:schema}

Let $\Delta x$ and $\Delta t^n$ be positive real numbers denoting the space and time increments, respectively. The mesh points $(x_{i},t^n)$ are characterized by $x_{i}=i\Delta x$ for every $i\in\mathbb{Z}$ and $t^{n+1}=t^n+\Delta t^n$ for every $n\in\mathbb{N}$. As is common in finite volume schemes \cite{leveque1992numerical}, we divide the $x$-axis into a sequence of cells $\{C_i\}_{i\in\mathbb{Z}}$ such that $C_i=[x_{i-1/2},x_{i+1/2})$, where $x_{i+1/2}=x_i+\Delta x/2$ are the cell interfaces. For every $n\in\mathbb{N}$ we construct a piecewise constant approximation $U^n_{j,i}$ of the variables $U_j = (\rho_j, v_j)^T$ for each lane $j=1, \dots, J$. These values represent the cell averages defined as:
\begin{equation*}
U^n_{j,i} = \frac{1}{\Delta x} \int_{C_i} U_j(x, t^n)\, \text{d}x.
\end{equation*}

The time step $\Delta t^n$ is adaptively chosen at each step $n$ to satisfy the CFL condition $\lambda_{\text{max}}^n \Delta t^n \le C \Delta x$ (where $C \leqslant 1$ is the CFL coefficient). Here, $\lambda_{\text{max}}^n$ denotes the largest wave speed \eqref{auto_tutto} in absolute value, which is time-dependent and therefore it is recomputed at each step $n$ based on the current values of the macroscopic variables $U^n_{j,i}$.

To address the fact that models 1 and 2 tend to become stiff when the parameter $\alpha$ is large, we adopt a implicit–explicit finite volume scheme (IMEX) \cite{leveque2002finite,pareschi2005implicit,crandall1980method}. In this approach, the nonlinear hyperbolic terms and the source terms related to the lane changes are handled explicitly, while the relaxation terms are treated implicitly. This strategy helps to effectively manage stiffness and to prevent the time-step restriction $\Delta t^n=O(1/\alpha)$. In particular, the scheme we will use is equivalent to a fractional step method, where the homogeneous part is initially approximated using a forward Euler step, followed by a backward Euler step to handle the relaxation part.

The update formula of a state $U^n_{j,i}$ in cell $C_i$ from time $t^n$ to the state $U^{n+1}_{j,i}$ at time $t^{n+1}$ is given by
\begin{equation*}
\begin{array}{ll}
\text{\underline{Model 1}} &  U^{n+1}_{j,i}=\tilde{U}^n_{j,i}+ \Delta t^n\begin{pmatrix} 
\nu S^n_{j,i} \\ 
\alpha\left(V_j(\rho_{j,i}^{n+1})-v_{j,i}^{n+1}+\frac{v_{j,i}^{n+1}}{v_{j,i}^n}R_{j,i}^n\right) \\ 
\end{pmatrix} \\
\text{\underline{Model 2}} &  U^{n+1}_{j,i}=\tilde{U}^n_{j,i}+ \Delta t^n\begin{pmatrix} 
\nu S^n_i \\ 
\alpha \left(V_j\left(\rho_{j,i}^{n+1}\right) - v_{j,i}^{n+1}\right) + \frac{Q^n_{j,i}}{\rho_{j,i}^n}-\nu S_{i,j}^n \left(P'(\rho_{j,i}^n)+\frac{v_{j,i}^n+P(\rho_{j,i}^n)}{\rho_{j,i}^n}\right) \\ 
\end{pmatrix}   \\
\end{array}
\end{equation*}where $\tilde{U}^n_{j,i}=(\tilde{\rho}^n_{j,i},\tilde{v}_{j,i}^n)^T$ denotes the  solution of the homogeneous hyperbolic part.
 
Note that in the second equation of Model 1, we have applied a semi-linearization technique \cite{pareschi2005implicit, patankar2018numerical} to the term \( R_j \). Specifically, we have rewritten \( R_j \) as \( R_j = \frac{v_j}{v_j} R_j \) and then approximated the velocity in the numerator using its value at time step \( n+1 \), while the velocity in the denominator and the term $R_j$ are approximated using their value at time step \( n \). This approach allows us to overcome the nonlinearity of the term $R_j$ appearing in the relaxation part. Note that with this choice, we restrict our simulations to the case of non-zero velocities. 

In order to compute $\tilde{U}^n_{j,i}$ we adopt a Roe-type wave-propagation scheme, that is an Upwind scheme by linearizing the Jacobian matrix at each cell interface developed by Roe \cite{roe1,roe2}. The scheme can be expressed by:
\begin{equation*}
\tilde{U}^{n}_{j,i}=U^n_{j,i}-\frac{\Delta t^n}{\Delta x} \left( (A^+(\Delta U_j))^n_{i-\frac{1}{2}}+ (A^-(\Delta U_j))^n_{i+\frac{1}{2}}\right).
\end{equation*}
The matrix $A(U_j)$ is linearized at cell interface using the mean $\bar{\rho}_j=\frac{1}{2}(\rho_{j,i}+\rho_{j,i+1})$ and $\bar{v}_j=\frac{1}{2}(v_{j,i}+v_{j,i+1})$. The solution gradient is projected on the characteristic fields of the linearized system \begin{equation}\label{univocamente} \Delta U_j = \sum_k \zeta_k \bar{r}_k,
\end{equation} for $k=1,\dots,2J$ yielding
\begin{equation*}
A^\pm (\Delta U_j) = \sum_{\lambda_k \gtrless 0} \zeta_k \bar{\lambda}_k \bar{r}_k,
\end{equation*}
where $\bar{r}_k$ and $\bar{\lambda}_k$ are the eigenvectors and eigenvalues of the local linearized matrix $A(U_j)$ at $x_{i+\frac{1}{2}}$, see \eqref{auto_tutto}, with the wave strengths \(\zeta_k\) uniquely determined by equation~\eqref{univocamente}. The explicit expression of the weights~$\zeta_k$ can be found, e.g., in~\cite{SongKarni}. We remind that both Model 1 and Model 2 share the same hyperbolic part of the model in~\cite{SongKarni}.

Since we use an implicit method only for the relaxation part in the velocity equation, the time update of the density variable is fully explicit. Therefore, in the time update of the velocity the quantity $\rho^{n+1}_i$ is already known. In general, using an implicit scheme, the new state $v^{n+1}_i$ is obtained as the solution of a scalar nonlinear equation, but in the particular case of our models, this update is given by an explicit formula. Therefore, we have
\begin{equation*}
\begin{array}{ll}
\text{\underline{Model 1}} &  v^{n+1}_{j,i}=\dfrac{\tilde{v}_{j,i}^n+\Delta t^n \alpha V_j(\rho_{j,i}^{n+1})}{1+\Delta t^n \alpha \left( 1-\frac{R_i^n}{v^n_{j,i}} \right)}  \\
\text{\underline{Model 2}} &  v^{n+1}_{j,i}=\dfrac{\tilde{v}^n_{j,i}+\Delta t^n \left( \alpha V_j(\rho_{j,i}^{n+1})+  \frac{Q}{\rho_{j,i}^n}-\nu S_i^n \left(P'(\rho_{j,i}^n)+\frac{v_{j,i}^n+P(\rho_{j,i}^n)}{\rho_{j,i}^n}\right) \right)}{1+\Delta t^n \alpha }   \\
\end{array}
\end{equation*}

Boundary conditions are imposed by introducing ghost cells adjacent to the outermost grid cells on both sides of the computational domain. For the borders, either free-flow conditions or periodic conditions are considered, as detailed in the numerical test section. This completes the description of the numerical scheme used for the following simulations.

\subsection{Test 1: asymptotic consistency between the microscopic and the macroscopic multilane models}

\begin{table}[t]
\centering
\resizebox{\columnwidth}{!}{
	\renewcommand{\arraystretch}{1.5}
\begin{tabular}{l|cc|cc}
& \multicolumn{2}{c|}{Lane 1} & \multicolumn{2}{c}{Lane 2} \\
& $t=0$ & $t=T=1000$ & $t=0$ & $t=T=1000$ \\
\hline \hline
\bf TEST 1a & & & &\\
\hline
\multirow{2}{*}{Micro Model \eqref{bftl_nlane}} & $N_1(0)=150$ & $N_1(T)=99$ & $N_2(0)=30$ & $N_2(T)=81$ \\
& $\rho_1^{(n)}(0)=1$ & $\rho_1^{(n)}(T)=0.66$ & $\rho_2^{(n)}(0)=0.20$ & $\rho_2^{(n)}(T)=0.54$ \\
& $v_1^{(n)}(0)=0$ & $v_1^{(n)}(T)=0.24$ & $v_2^{(n)}(0)=0.80$ & $v_2^{(n)}(T)=0.47$ \\
\hline
Macro Model 1 \eqref{eq:model1:noncons} & $\rho_{1}(x,0)=1$ & $\rho_{1}(x,T)=0.70$ & $\rho_{2}(x,0)=0.20$ & $\rho_{2}(x,T)=0.50$ \\
& $v_{1}(x,0)=0$ & $v_{1}(x,T)=0.21$ & $v_{2}(x,0)=0.80$ & $v_{2}(x,T)=0.50$\\
\hline
Macro Model 2 \eqref{eq:model2:noncons} & $\rho_{1}(x,0)=1$ & $\rho_{1}(x,T)=0.70$ & $\rho_{2}(x,0)=0.20$ & $\rho_{2}(x,T)=0.50$ \\
& $v_{1}(x,0)=0$ & $v_{1}(x,T)=0.21$ & $v_{2}(x,0)=0.80$ & $v_{2}(x,T)=0.50$ \\
\hline\hline

\bf TEST 1b & & & &\\
\hline
\multirow{2}{*}{Micro Model \eqref{bftl_nlane}} & $N_1(0)=1000$ & $N_1(T)=67$ & $N_2(0)=50$ & $N_2(T)=83$ \\
& $\rho_1^{(n)}(0)=0.67$ & $\rho_1^{(n)}(T)=0.45$ & $\rho_2^{(n)}(0)=0.33$ & $\rho_2^{(n)}(T)=0.55$ \\
& $v_1^{(n)}(0)=0.23$ & $v_1^{(n)}(T)=0.39$ & $v_2^{(n)}(0)=0.67$ & $v_2^{(n)}(T)=0.45$ \\
\hline
Macro Model 1 \eqref{eq:model1:noncons} & $\rho_{1}(x,0)=0.67$ & $\rho_{1}(x,T)=0.50$ & $\rho_{2}(x,0)=0.33$ & $\rho_{2}(x,T)=0.50$ \\
& $v_{1}(x,0)=0.23$ & $v_{1}(x,T)=0.35$ & $v_{2}(x,0)=0.67$ & $v_{2}(x,T)=0.50$\\
\hline
Macro Model 2 \eqref{eq:model2:noncons} &  $\rho_{1}(x,0)=0.67$ & $\rho_{1}(x,T)=0.50$ & $\rho_{2}(x,0)=0.33$ & $\rho_{2}(x,T)=0.50$ \\
& $v_{1}(x,0)=0.23$ & $v_{1}(x,T)=0.35$ & $v_{2}(x,0)=0.67$ & $v_{2}(x,T)=0.50$ \\
\hline
\end{tabular}}\caption{Test 1. Numerical experiments assessing the consistency of the macroscopic limit for the microscopic and macroscopic multilane models. The table shows the initial conditions  and the final states  for both the microscopic model and two macroscopic models at the final time $T=1000$. The results indicate a strong agreement between the microscopic and macroscopic models, demonstrating that the macroscopic models accurately replicate the behavior of the microscopic model.\label{tab:consi}}
\end{table}

The following tests aim to analyze the consistency between the solutions of the BFtL multilane model \eqref{bftl_nlane}, with $K=0$, and the macroscopic multilane models \eqref{eq:model1:noncons} and \eqref{eq:model2:noncons}, with the pressure~\eqref{defP}. The microscopic model is solved using a fourth-order Runge-Kutta method.

We consider a circular two-lane road, which means that we use periodic boundary conditions, with a total length of $\tilde{L} = 1500$ meters, where each vehicle has a fixed length of $\tilde{l} = 5$ meters. In the numerical simulations, we employ normalized parameters, specifically $L = 1$, $l = d_s = 0.00\overline{3}$, and a final time of $T = 1000$.

For the initial conditions, we set the number of vehicles in each lane such that they are uniformly distributed. This configuration ensures local equilibrium within each lane while still allowing for lane changes.

For the microscopic model, after selecting the values of $N_j(0)$ for $j=1,2$, we impose the following initial conditions:
\begin{equation*}
\begin{array}{l}
x_{s_n^j}(0)-x_n(0)=\frac{L}{N_j(0)}, \\ v_n(0)=\tilde{V}_j\left(\frac{L}{N_j(0)}\right),
\end{array} \quad \forall n\in I_j(0).
\end{equation*}

For the macroscopic models, the initial conditions are derived from the corresponding initial local discrete densities defined in~\eqref{piece}, leading to:
\begin{equation*}
\begin{array}{l}
\rho_j(x,0)=\sum_{n\in I_j(0)}\rho_j^{(n)}(t=0) \chi_{[x_n(0),x_{s_n^j}(0))}(x), \\ v_j(x,0)=V_j(\rho_j(x,0)),
\end{array}
\end{equation*}
where $\chi_I(x)$ is the characteristic function of the set $I \subset \mathbb{R}$. Furthermore, $\Delta x = 0.001$.

In addition, we choose $v_2^{\max}=1$ and $v_1^{\max}=0.7$, and the other model parameters are set as $\alpha=\beta=1$, $\gamma=2$, and $\eta=0$. Note that $V(r)=\tilde{V}\left(\frac{l+d_s}{r}\right)$. 

The initial conditions for Test 1a and Test 1b are specifically chosen to analyze different levels of congestion in the slow lane (lane 1) while lane 2 is kept relatively free. In Test 1a, lane 1 is totally congested, while in Test 1b, it is only partially congested. These setups are designed to activate lane-changing maneuvers toward the faster lane, allowing for a robust verification of the micro-to-macro consistency of the source terms.

The results of the tests are summarized in Table~\ref{tab:consi}, where we report the selected initial conditions, including the number of vehicles and the local densities and velocities for the microscopic model, as well as the macroscopic density and velocity for the continuum models. We also report the final states obtained through the numerical evolution of the models.

Both simulations show good agreement between the microscopic multilane model and its macroscopic counterparts at the final time.

\subsection{Test 2: clearing a traffic queue and comparison with the first-order multilane model}

\begin{figure}[t]
\centering
\includegraphics[width=1.1\textwidth]{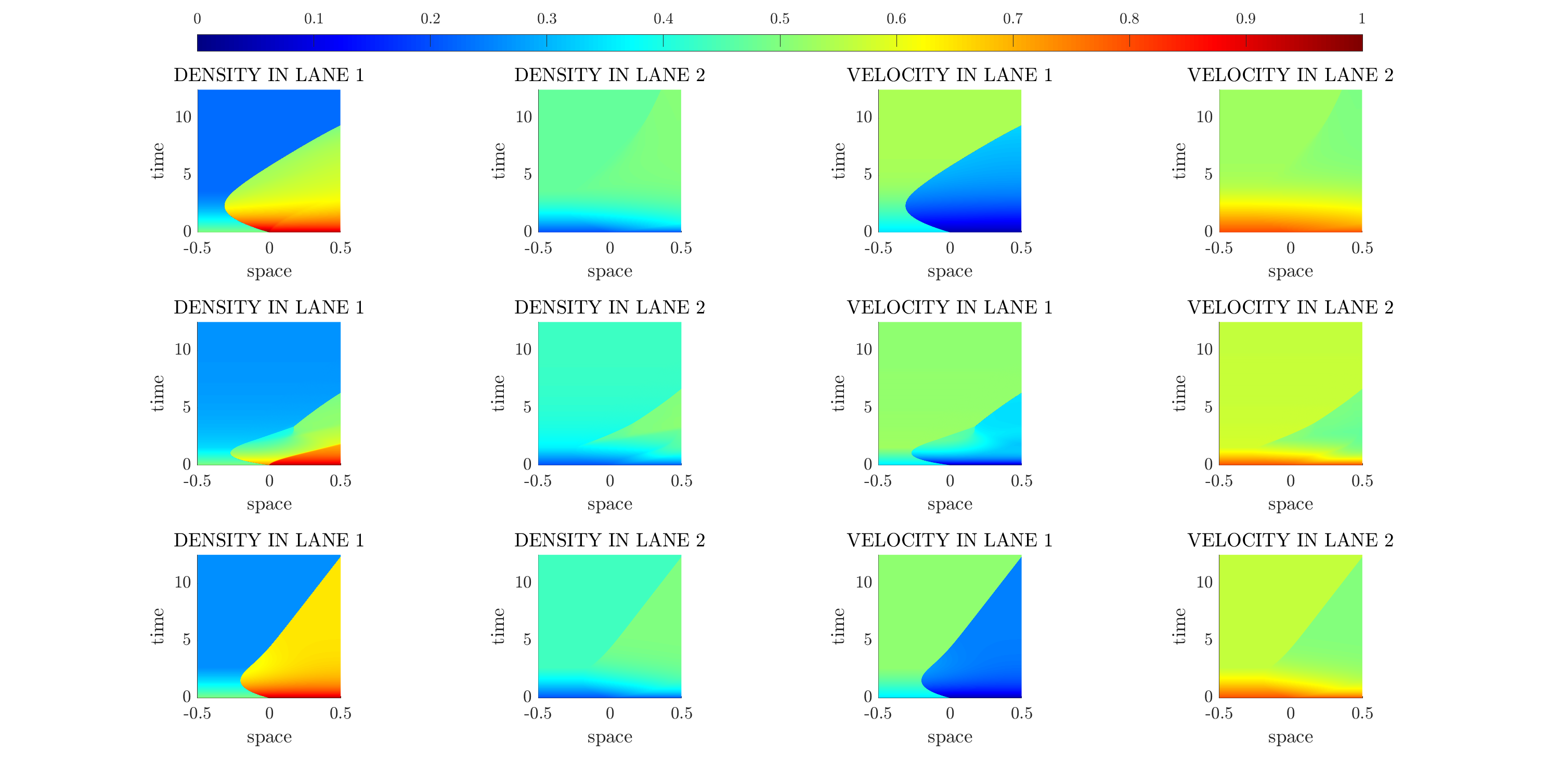}
\caption{Test 2: clearing a traffic queue and comparison with the first-order multilane model. The figure shows four panels in each row. In the first two panels, the time evolution of the density in lanes 1 and 2 is displayed, while the last two panels show the time evolution of the velocity in lanes 1 and 2.  The first row presents the results of the second-order model~\eqref{eq:model1:noncons}, the second row shows the corresponding results of the second-order model~\eqref{eq:model2:noncons}, and the third row contains the results of the first-order model~\eqref{macro_1_ord}.}
\label{fig:coda}
\end{figure}

In this test, we study the dynamics of the second-order models~\eqref{eq:model1:noncons} and~\eqref{eq:model2:noncons}, with the pressure~\eqref{defP}, in a case of a traffic jam. We also compare the results with the predictions provided by the first-order macroscopic model~\eqref{macro_1_ord}, which is solved using a finite volume scheme with the Rusanov numerical flux.

We consider a two-lane road represented by the interval $[-0.5,0.5]$, with free flow conditions at the border and with $v_2^{\max}=1$ and $v_1^{\max}=0.7$. The initial conditions are chosen such that, in lane 1, there is a queue of vehicles moving slowly on the right side of the domain, while in lane 2, the vehicle density is spatially low, allowing vehicles to flow smoothly. To ensure a consistent comparison between the second-order models and the first-order model, we initialize the velocities using the equilibrium speed values.

Specifically, the initial data are defined as follows:
\begin{equation*}
\begin{array}{ll}
\rho_1(x,0) = \begin{cases} 
0.50 & x \in [-0.5, 0), \\ 
0.95 & x \in [0, 0.5], 
\end{cases} 
& v_1(x,0) = V_1(\rho_1(x,0)) = \begin{cases} 
0.350 & x \in [-0.5, 0), \\ 
0.035 & x \in [0, 0.5], 
\end{cases} \\[10pt]
\rho_2(x,0) = 0.2 \quad \forall x \in [-0.5, 0.5], 
& v_2(x,0) = V_2(\rho_2(x,0)) = 0.80 \quad \forall x \in [-0.5, 0.5].
\end{array}
\end{equation*}

The model parameters are: $\alpha = \nu = 1$, $l = d_s = 0.5$, $\beta = \gamma = 2$, $\eta = 0.1$. Additionally, $\Delta x = 0.001$ and $T = 12.5$.

The results are illustrated in Fig.~\ref{fig:coda}. Lane-changing is activated promptly in all models, facilitating the migration of certain vehicles from lane 1 to lane 2. This transition effectively alleviates the traffic congestion present in lane 1 while concurrently increasing the vehicle density in lane 2, which leads to a corresponding decrease in average speed. The qualitative trends exhibited by the three models are comparable. However, Model 2 predicts a more rapid dissipation of the queue compared to Model 1. Furthermore, it is noteworthy that in the first-order model, the queue dissipates more gradually due to the direct correlation between velocity and density, which results in a less dynamic response to changing traffic conditions. Indeed, in the first-order model, non-equilibrium phenomena cannot be accurately represented.

\begin{figure}[htbp]
\centering
\begin{subfigure}[htbp]{\textwidth}
\centering
\includegraphics[width=\textwidth]{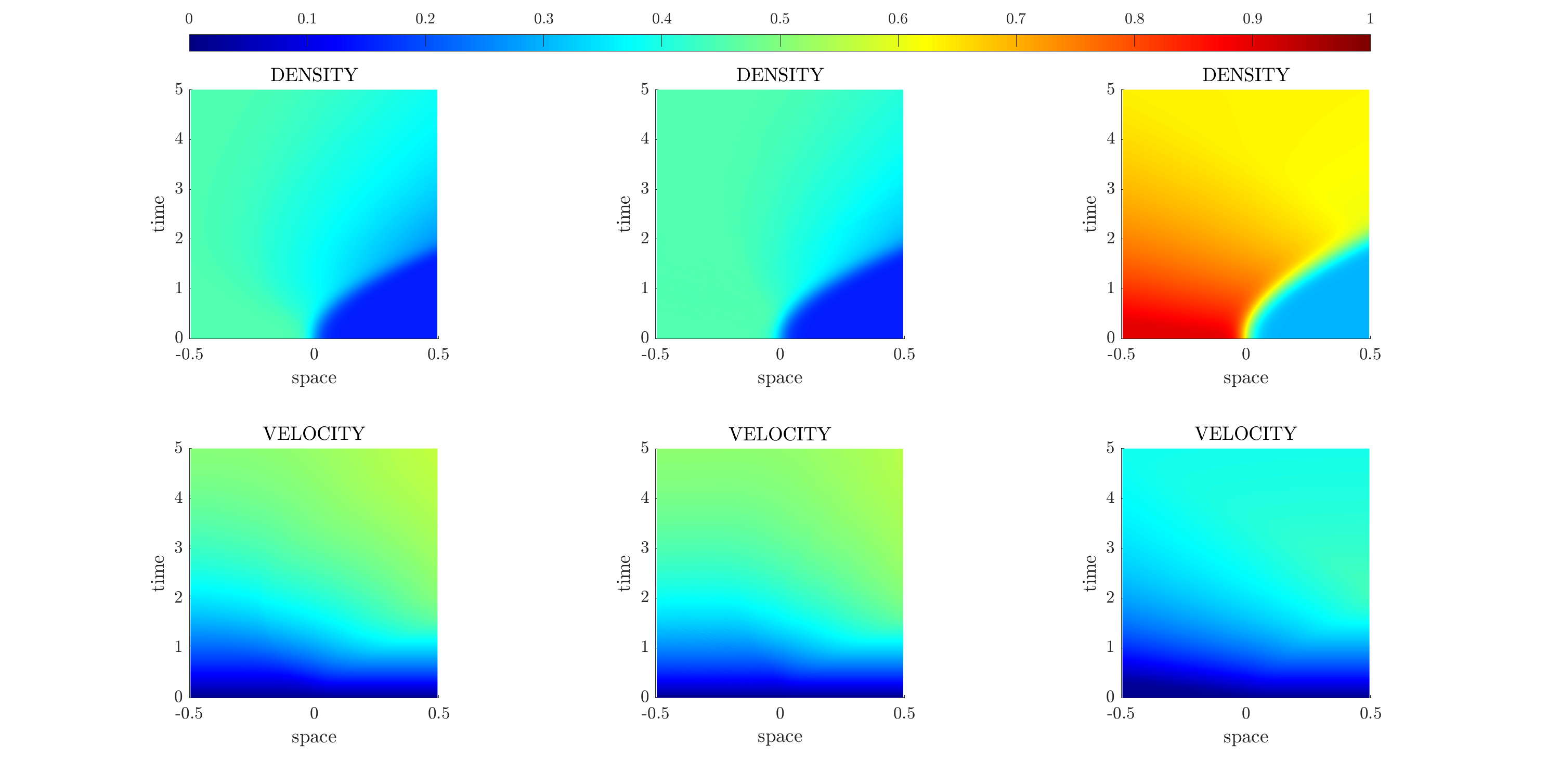}
\caption{The first column displays the evolution of the average density $\frac{1}{2}(\rho_1 + \rho_2)$ and average velocity $\frac{1}{2}(v_1 + v_2)$ for Model 1~\eqref{eq:model1:noncons}. The second column presents the corresponding average quantities for Model 2~\eqref{eq:model2:noncons}. The third column shows the evolution of density and velocity for the single-lane ARZ model~\eqref{eq:arz}. This aggregate representation facilitates a direct comparison of how lane-changing dynamics in multilane models regularize traffic flow compared to the single-lane scenario.}
\label{fig:testmedie1}
\end{subfigure}
\begin{subfigure}[htbp]{\textwidth}
\centering
\includegraphics[width=\textwidth]{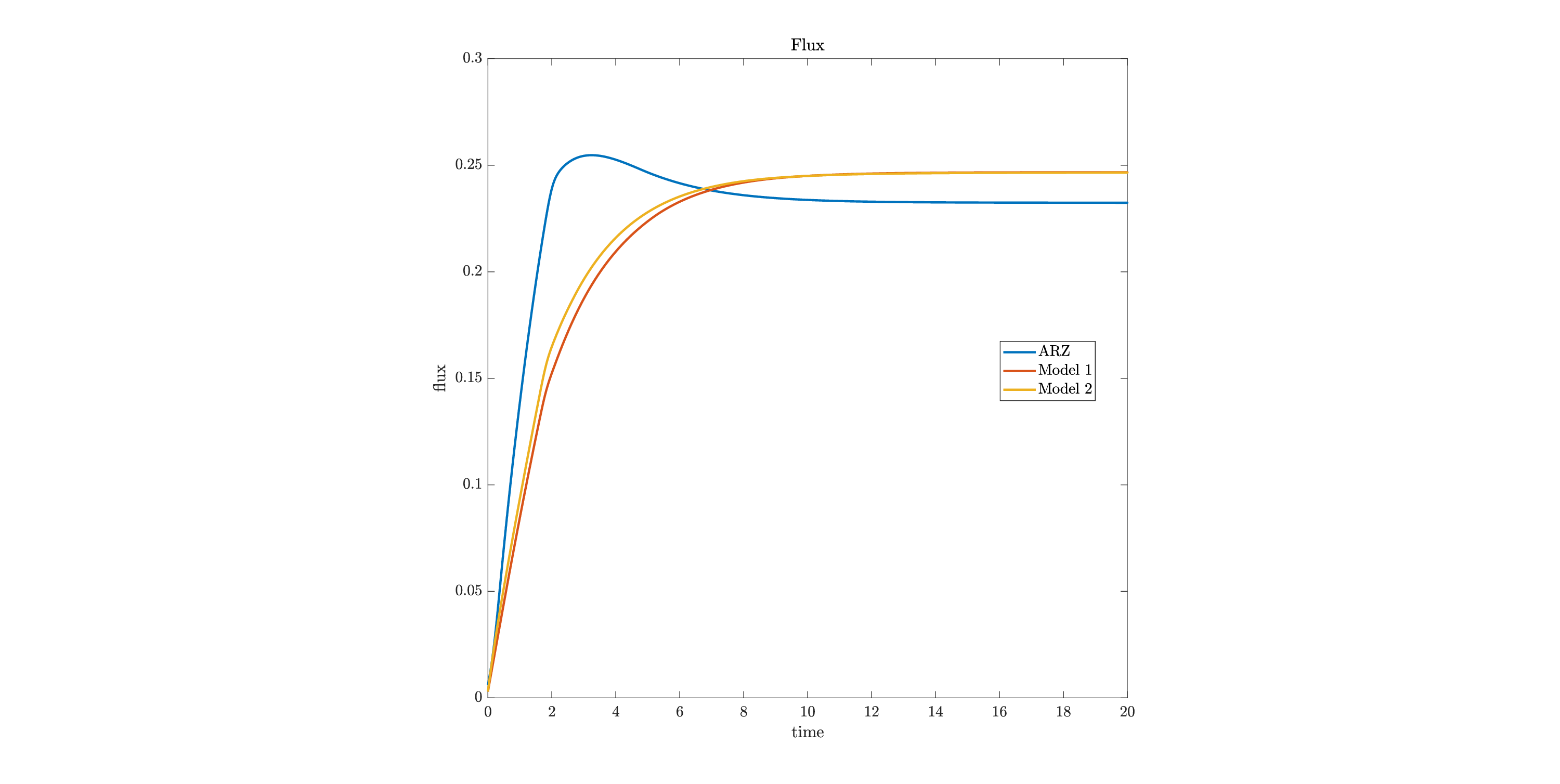}
{\caption{Time evolution of the average flow \( F(t) \) for the single-lane model and the mean flow \( \frac{1}{2}(F_1(t)+F_2(t)) \) for the multi-lane models.}}
\label{fig:testmedie2}
\end{subfigure}
\caption{Test 3. Traffic regularization and comparison with the single-lane case}
\label{fig:testmedie}
\end{figure}
\subsection{Test 3: traffic regularization and comparison with the single-lane case}
In this test, we consider a road with \( J = 2 \) lanes, where the lane 1 exhibits high traffic density in the left part of the spatial domain and lower traffic density in the right part. The dynamics of the second-order models~\eqref{eq:model1:noncons} and~\eqref{eq:model2:noncons}, with pressure defined as in~\eqref{defP}, are compared with the single-lane ARZ model~\eqref{eq:arz}, where the initial mass is assumed to be equal to the sum of the masses in the two-lane version.

Specifically, our computational domain is given by \([-0.5,0.5]\) with a final time \( T=20 \). We employ a regularization of the step-like initial data using the sigmoid function \( s(x)=(1+\exp{(-50x)})^{-1} \) to obtain solutions that satisfy the entropy condition under the Roe scheme used in our simulations.

The initial conditions for both two-lane models are defined for \( j=1,2 \) as follows: \begin{equation*}
\rho_j(x,0)=\rho_j^L(1-s(x))+\rho_j^Rs(x), \quad v_j(x,0)=0.01,
\end{equation*} with \(\rho_1^L=0.7, \rho_1^R=0.1\) and \(\rho_2^L=0.2, \rho_2^R=0.2\).

For the single-lane model, we consider: \begin{equation*}
\rho(x,0)=(\rho_1^L+\rho_2^L)(1-s(x))+(\rho_1^R+\rho_2^R)s(x), \quad v(x,0)=0.01.
\end{equation*}

In all cases, we impose free-flow boundary conditions. The parameters of the models are set as follows: \(\alpha = 0.5\), \(\beta = \nu = 1\), \(l = d_s = 0.5\), \(\gamma = 2\), and \(\eta = 0.1\). Additionally, the spatial step size is \(\Delta x = 0.001\).

The results regarding the evolution of density and velocity are presented in Fig.~\ref{fig:testmedie1}, where, for visualization purposes, we display the solution at the intermediate time \( \tilde{T} = 5 \). To provide a more comprehensive comparison with the single-lane scenario, we have included  the evolution of the average density $ \frac{1}{2}(\rho_1(x,t) + \rho_2(x,t))$ and velocity  $ \frac{1}{2}(v_1(x,t) + v_2(x,t))$ for both multilane models. This allows one to appreciate how the mass redistribution facilitated by lane-changing maneuvers affects the overall traffic flow regularization compared to the standard ARZ model.

We observe that in the multilane scenario, the activation of lane changes, leads to an increase in the average velocity in both lanes compared to the single-lane case. In the latter, where the same total mass is considered, the average velocities remain lower.

This qualitative behavior is further confirmed by the explicit computation of the average flow over time, given by \begin{equation}\label{calcoloQ}
F_j(t) = \int_{-0.5}^{0.5} \rho_j(x,t) v_j(x,t) \,\text{d}x, \quad j=1,2.
\end{equation} In Fig.~\ref{fig:testmedie}, we compare the flow \( F(t) \) for the single-lane ARZ model~\eqref{eq:arz} with the average of the flows in the two-lane models, given by \( \frac{1}{2}(F_1(t)+F_2(t)) \) for Models 1 and 2. The integral is numerically approximated using the trapezoidal rule. In both multilane scenarios, the obtained average flow is higher than that of the single-lane case.

\begin{figure}[htbp]
\centering
\begin{subfigure}{\textwidth}
\centering
\includegraphics[width=1\textwidth]{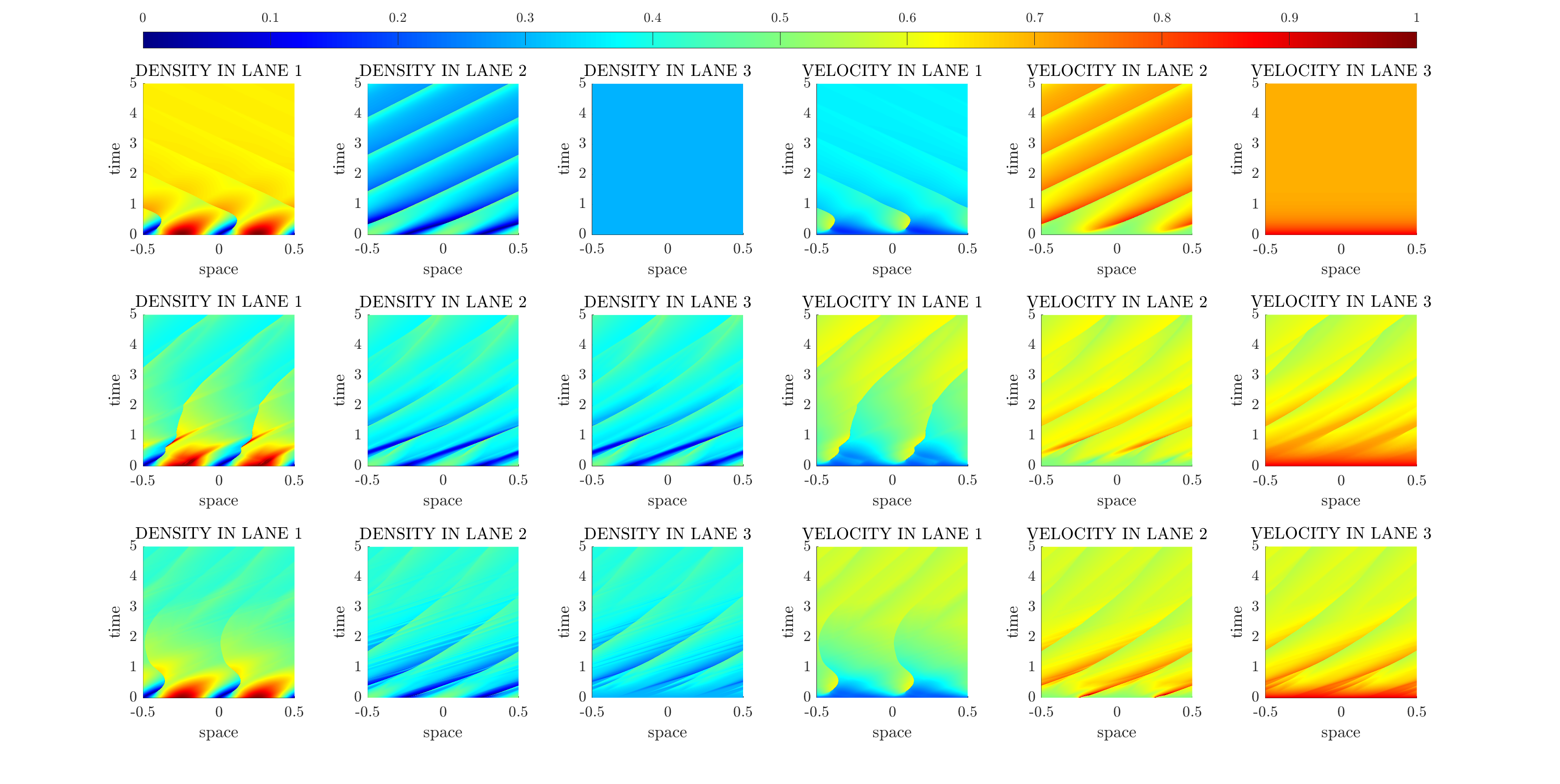}
\caption{The first row shows the evolution of the density and velocity in each lane using the Aw-Rascle and Zhang model~\eqref{eq:arz}, without lane changes. The second row presents the corresponding density and velocity profiles obtained with Model 1~\eqref{eq:model1:noncons}. The third row displays the density and velocity for each lane using Model 2~\eqref{eq:model2:noncons}.}
\label{fig:3lanes}
\end{subfigure}

\vspace{0.5cm}

\begin{subfigure}{\textwidth}
\centering
\includegraphics[width=0.9\textwidth]{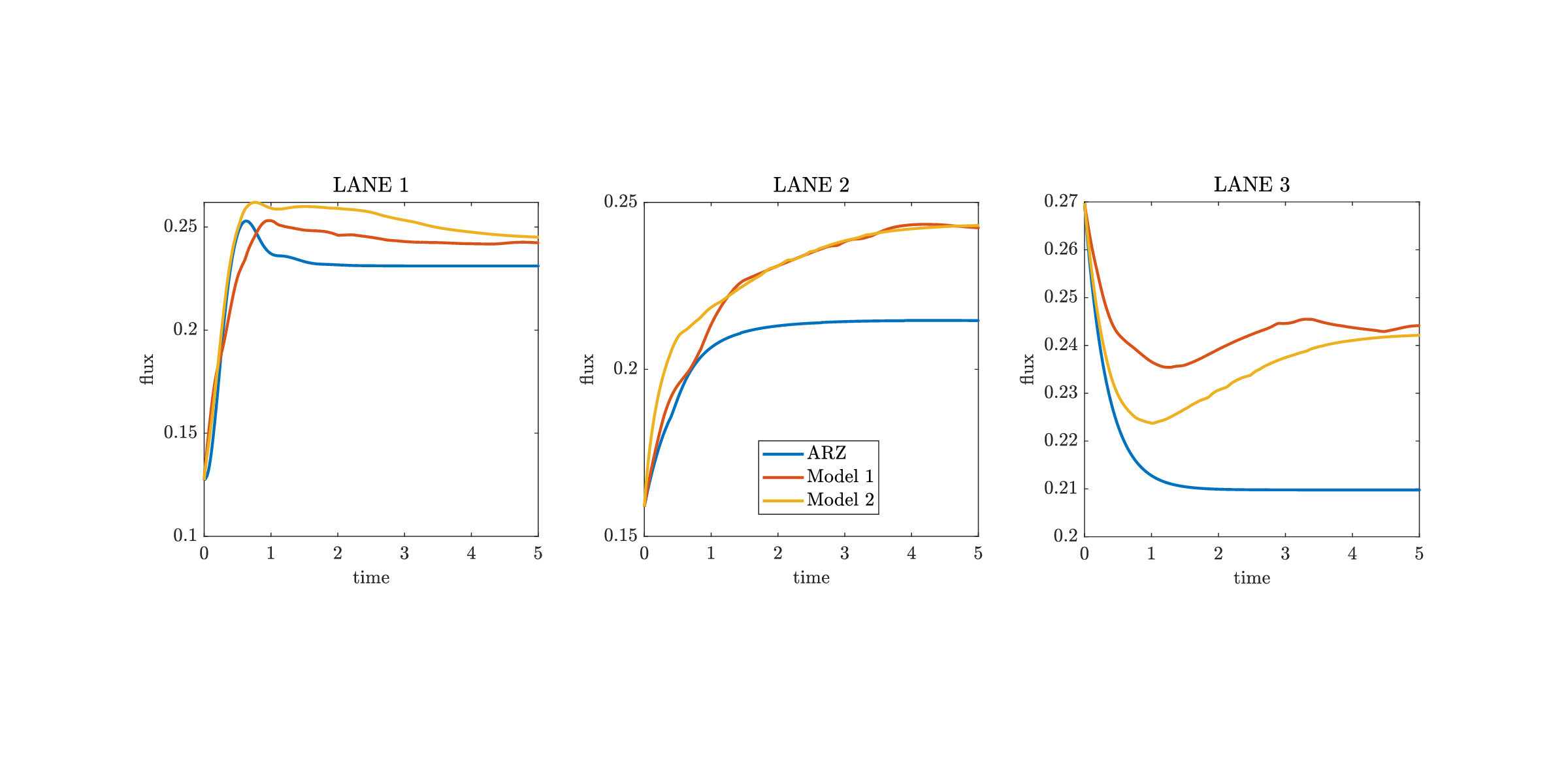}
\caption{Time evolution of the average flow \( F_j(t) \) for each lane (\( j = 1,2,3 \)). Each subplot represents a different lane, showing the comparison between the three models.}
\label{fig:3lanesflussi}
\end{subfigure}
\caption{Test 4. A three-lane case and comparison of scenarios with and without lane changes. \emph{(Continued on next page)}}
\end{figure} 

\begin{figure}[htbp]
\ContinuedFloat 
\centering
\begin{subfigure}{\textwidth}
\centering
\includegraphics[width=0.9\textwidth]{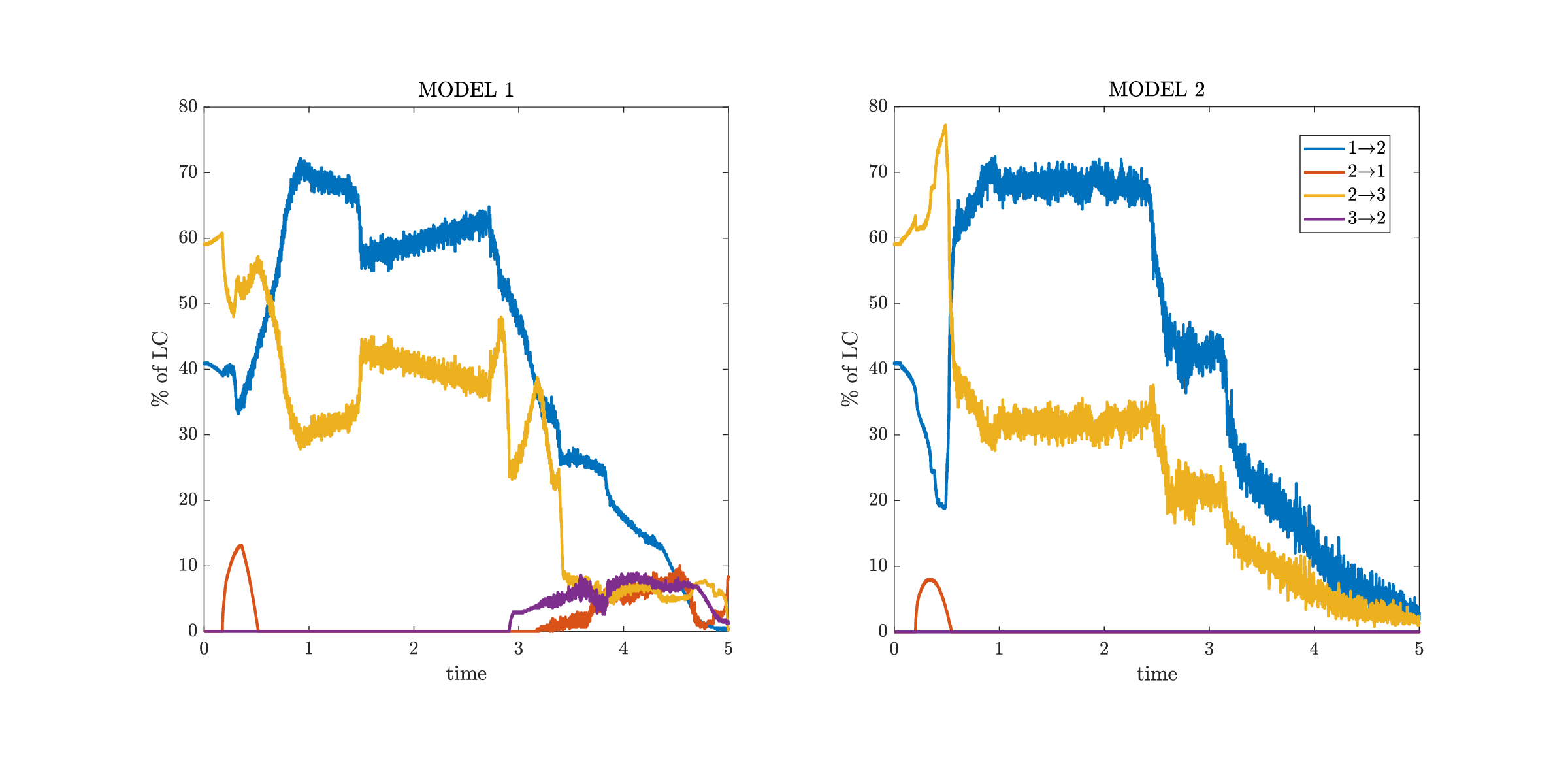}
\caption{Lane change percentages over time for both models. The left panel shows the results for Model 1, while the right panel displays the results for Model 2. The graphs illustrate the evolution of lane changes in terms of the percentage of computational cells affected over time.}
\label{fig:3lanescambi}
\end{subfigure}
\caption[]{Test 4. A three-lane case and comparison of scenarios with and without lane changes. \emph{(Continued)}}
\label{fig:test3}
\end{figure}

\subsection{Test 4: a three-lane case and comparison of scenarios with and without lane changes}
In this test, we analyze a three-lane road modeled by the spatial domain \([-0.5, 0.5]\), where the initial configuration of density and velocity in lane 1 gives rise to traveling perturbations of the flow. In this example, we employ identical profiles of desired speed, given by \eqref{lineari} with $v_j^\text{max}=1$, for all $j=1,2,3$. The objective of our investigation is to evaluate traffic dynamics in the absence and presence of lane-changing mechanisms. Specifically, we juxtapose the results derived from the Aw-Rascle and Zhang model \eqref{eq:arz}, which lacks lane-changing capabilities, with the results from two multilane models that account for lane-changing effects, namely~\eqref{eq:model1:noncons} and~\eqref{eq:model2:noncons}.

We use the following initial conditions for the traffic density and velocity in each lane:
\begin{equation*}
\begin{array}{ll}
\rho_1(x,0) = |\sin(2\pi x)|, & v_1(x,0) = 0.2, \\
\rho_2(x,0) = 0.5 |\cos(2\pi x)|, & v_2(x,0) = 0.5, \\
\rho_3(x,0) = 0.3, & v_3(x,0) = 0.9, \\
\end{array}
\end{equation*}
The parameters of the models are set as follows: \(\alpha = 3\), \(\beta = \nu = 1\), \(l = d_s = 0.5\), \(\gamma = 2\), and \(\eta = 0.1\). Additionally, the spatial step size is \(\Delta x = 0.001\) and the final simulation time is \(T = 5\).

The results are presented in Fig.~\ref{fig:3lanes}. It can be observed that the backward-propagating perturbations, which appear in the scenario without lane changes, are significantly dampened and reduced when lane-changing is enabled. This redistribution of vehicles between lanes leads to a more stable system, smoothing out traffic fluctuations. By computing the average flow over time for all three models, given by \eqref{calcoloQ}, we observe that enabling lane changes leads to a higher vehicle flow across all three lanes compared to the single-lane scenario, see Fig.~\ref{fig:3lanesflussi}. Additionally, in Fig.~\ref{fig:3lanescambi}, we report the number of lane changes over time, calculated as the percentage of computational cells in which lane changes occur. To do this, at each time step, we count how many cells experience lane changes and divide that by the total number of cells, thus obtaining a percentage. This allows us to track the dynamics of lane changes over time and better understand their distribution in the system. We notice that in both cases, there is a higher number of lane changes toward lane 3. Additionally, in Model 1, we observe mass lane changes originating from lane 3 to lane 2, which are completely absent in Model 2.


\subsection{Test 5: validation with experimental traffic data}

To assess the capability of the model to reproduce realistic traffic behavior, we present  numerical experiments based on empirical data from multilane highways. Specifically, we compare the fundamental diagrams derived from real measurements with those obtained from simulations of the second-order models \eqref{eq:model1:noncons} and \eqref{eq:model2:noncons}, where the parameters of the theoretical fundamental diagrams have been calibrated using the empirical data.

This comparison provides a qualitative validation step for the models: by verifying that the simulated flow-density and speed-density relationships resemble their empirical counterparts, particularly in terms of  lane-specific differences and capacity drop, we demonstrate the model's ability to capture key features of real-world traffic dynamics, even beyond the range directly observed in the data. This approach also highlights the advantage of second-order models in reproducing complex traffic phenomena that first-order models typically fail to describe, because they cannot provide multivalued fundamental diagrams.

\paragraph{Empirical datasets and calibration} We considered two distinct datasets that offer lane-resolved measurements of flow, density, and speed, enabling the construction of empirical fundamental diagrams. The first dataset consists of traffic flow measurements collected via fixed sensors along a two-lane segment of the Italian highway A4 Trieste–Venice, provided by Autovie Venete S.p.A.~\cite{avnorme,maya}. The second dataset is based on vehicle trajectories extracted from video recordings on a three-lane section of the German motorway A3 near Frankfurt am Main~\cite{kallo2019microscopic,2dimdata}. Regarding the Italian dataset, macroscopic traffic quantities were obtained from fixed inductive-loop detectors located along the A4 motorway. These sensors provide direct measurements of flow and average velocity, which are then aggregated every 60 seconds. The macroscopic density is then derived using the relation $\rho = f/v$ \cite{maya}. For light and heavy vehicles, data from both the slow and fast lanes are considered; however, it should be noted that for heavy vehicles, only the slow lane is actually involved, reflecting the specific overtaking restrictions in place on that highway segment.
For the German dataset, which provides individual vehicle trajectories, macroscopic quantities (density, flow, and velocity) were extracted following the procedure in \cite{2dimdata}. Microscopic velocities were first estimated from 2D positions via a least-squares linear approximation. Macroscopic values were then computed every second and subsequently aggregated over a 60-second time horizon. This averaging process filters out high-frequency fluctuations and allows for a consistent comparison with the Italian sensor data.

\paragraph{Fitting of the theoretical fundamental diagrams} For each dataset, we constructed empirical fundamental diagrams for each lane by plotting flow vs.~density and velocity vs.~density, using traffic measurements averaged over short spatial and temporal intervals. As a reference for fitting, we adopted the three-parameter family of smooth and strictly concave flow rate curves introduced in \cite{3param1,gsom1}, given by
\begin{equation}
\label{FD:3param}
\begin{split}
f_{\alpha,\lambda,p}(\rho) = \alpha\Bigg(
&\sqrt{1+(\lambda p)^2}
+ \left(\sqrt{1+(\lambda(1-p))^2} - \sqrt{1+(\lambda p)^2}\right)\frac{\rho}{\rho^{\max}}\\
&- \sqrt{1+\left( \lambda\left( \frac{\rho}{\rho^{\max}} - p \right) \right)^2}
\Bigg).
\end{split}
\end{equation} Each flow  function defined in \eqref{FD:3param} is zero at both \( \rho = 0 \) and \( \rho = \rho^{\text{max}} \). The three free parameters govern key features of \( f_{\alpha,\lambda,p}(\rho) \): the maximum flow rate \( f^{\text{max}} \) is primarily influenced by \( \alpha \); the critical density \( \sigma \), defined as the density at which the flow reaches its maximum, is mostly determined by \( p \); and the smoothness of the transition around \( \sigma \), which reflects how sharply the slope changes sign  is mainly controlled by \( \lambda \). This fundamental diagram was chosen for this test, because the presence of several parameters allows for a good fit between model and actual data

Given empirical measurements \((\tilde{\rho}_i, \tilde{f}_i)\) for \(i = 1, \dots, i^{\text{max}}\), the values of \(\sigma\) and \(f^{\text{max}}\) were identified by solving a nonlinear least-squares problem:
\begin{equation}\label{lls}
\min_{\alpha,\lambda,p} \sum_{i=1}^{i^{\text{max}}} \left|f_{\alpha,\lambda,p}(\tilde{\rho}_i) - \tilde{f}_i\right|^2.
\end{equation} In other words, the objective is to minimize the \(L^2\) norm between the observed flow values and the theoretical predictions of the three-parameter flux function \eqref{FD:3param}.

\begin{table}[ht]
\centering
\begin{tabular}{llccccc}
\hline
 & & \multicolumn{2}{c}{\textbf{Italian dataset}} & \multicolumn{3}{c}{\textbf{German dataset}} \\
\cmidrule(lr){3-4} \cmidrule(lr){5-7}
 & & Lane 1 & Lane 2 & Lane 1 & Lane 2 & Lane 3 \\
\hline
\multirow{3}{*}{Parameters}
& $\alpha$ (veh/sec)   & 495.30 & 862845.49 & 99.05 & 229.01 & 108.55 \\
& $\lambda$ & 0.20   & 3.65      & 0.09  & 0.13   & 0.10 \\
& $p$       & 7.63   & 1.02      & 74.03 & 36.05  & 79.48 \\
\hline
\end{tabular}
\caption{Test 5. Parameters from data fitting for Italian and German datasets.}
\label{tab:params_sub}
\end{table}

The solutions of \eqref{lls} obtained for the two datasets are summarized in Table~\ref{tab:params_sub}, while the corresponding fitted fundamental diagrams are shown in Figs.~\ref{fig:realdata_italy} and~\ref{fig:realdata_germany}, where they are overlaid with the empirical measurements for visual comparison. It is worth recalling that the value of the maximum density \(\rho^{\text{max}}\) typically reflects the physical characteristics of the road rather than being inferred directly from the data. For the Italian dataset, which includes \(i^{\text{max}} = 2880\) measurements, we selected a value of \(\rho^{\text{max}}\) equal to 50 veh/km for lane 1 and 90 veh/km for lane 2, consistent with a heterogeneous traffic composition involving both heavy and light vehicles. Notably, lane 1 experiences a significant flow of trucks and articulated lorries, whose lengths are considerably greater than those of standard passenger cars. In contrast, for the German dataset, comprising \(i^{\text{max}} = 601\) measurements, a higher value of \(\rho^{\text{max}}=133\) veh/km, was adopted to reflect a traffic stream composed almost exclusively of passenger vehicles for all lanes.

\begin{figure}[ht]
\centering
\includegraphics[width=\textwidth]{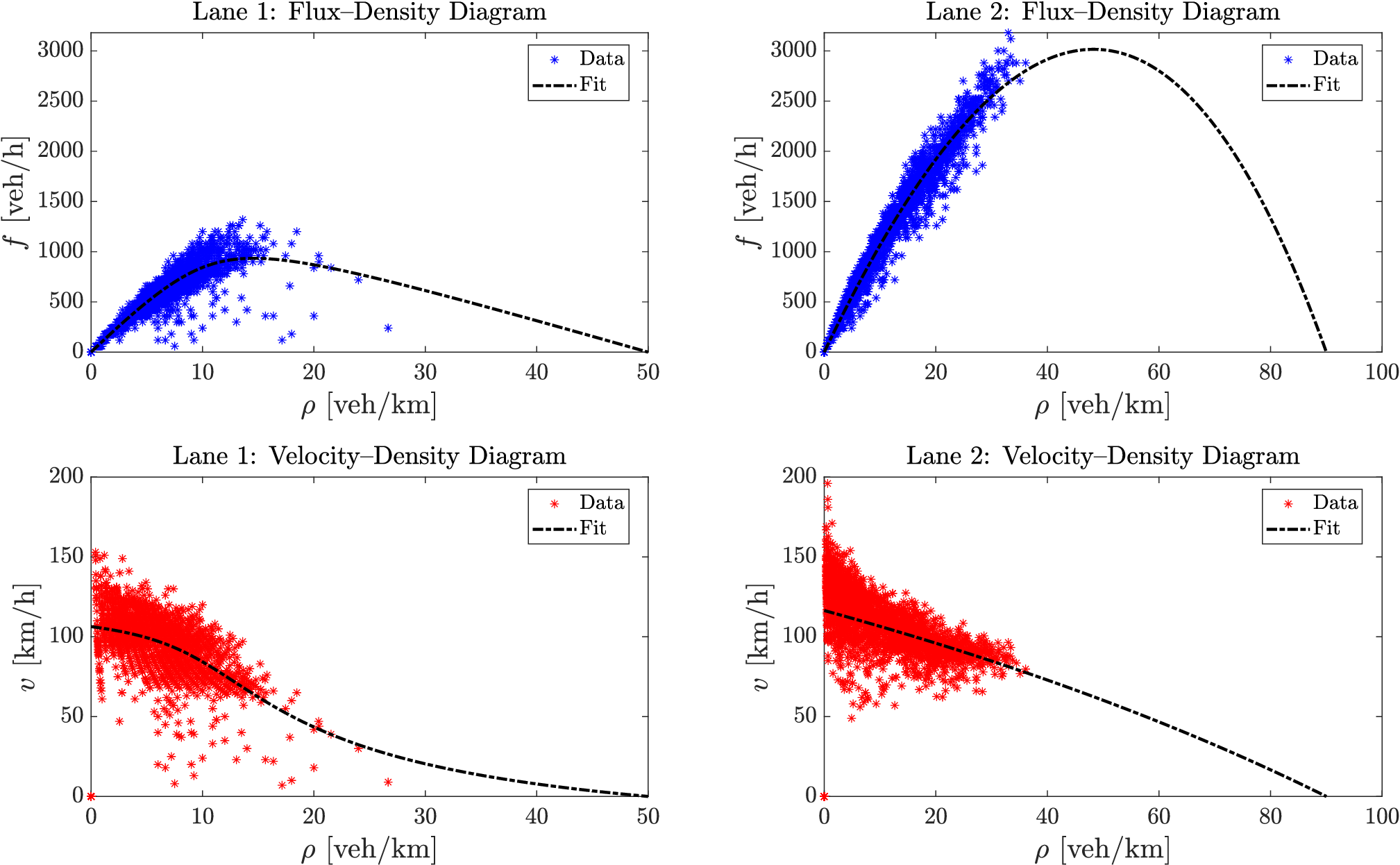}
\caption{Empirical data from the Italian dataset. The top row shows density--flow diagrams, and the bottom row shows density--velocity diagrams. Data from the two lanes of the A4 highway segment are displayed, with theoretical fits superimposed.}
\label{fig:realdata_italy}
\end{figure}

\begin{figure}[ht]
\centering
\includegraphics[width=\textwidth]{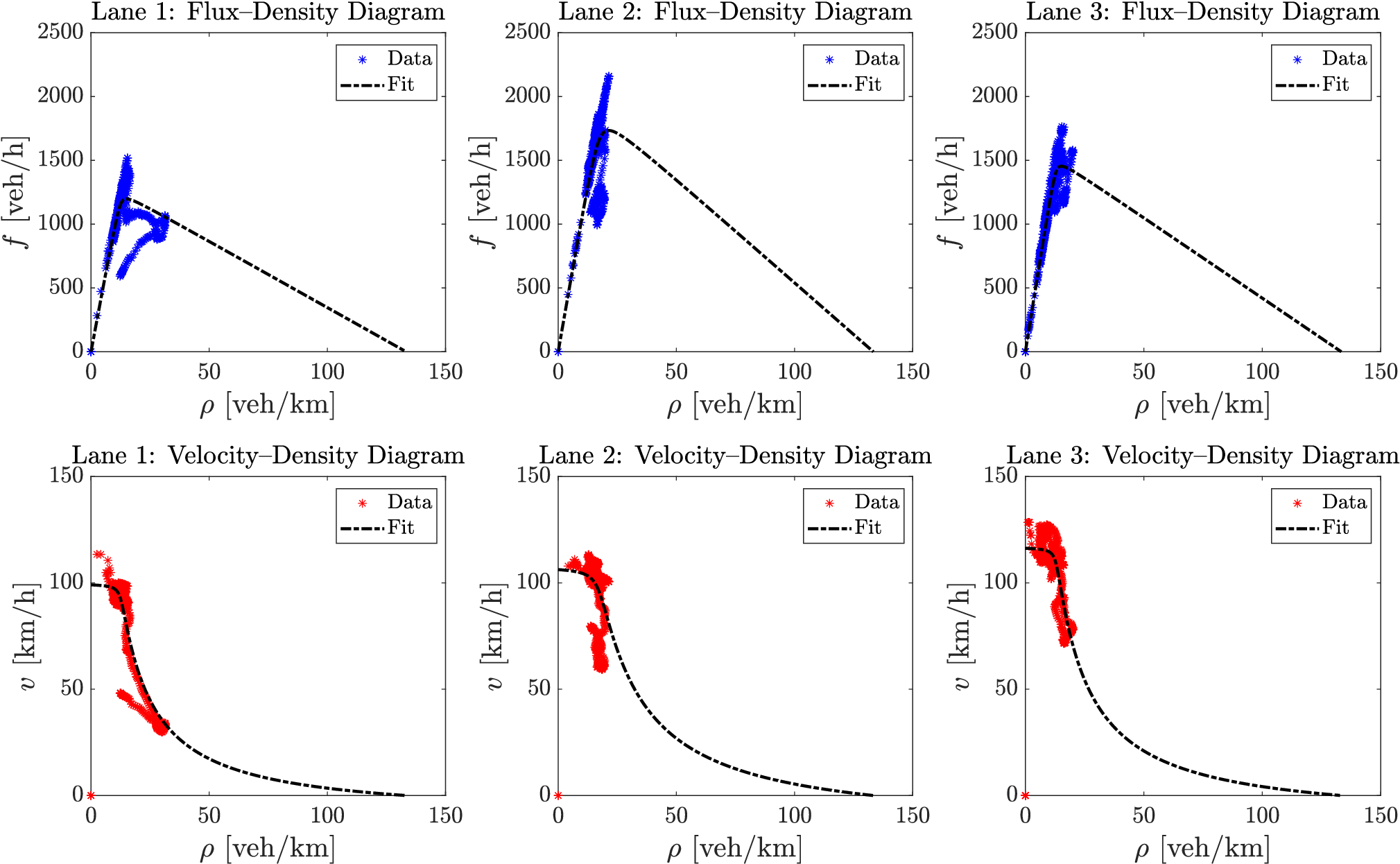}
\caption{Test 5. Empirical data from the German dataset. The top row shows density--flow diagrams, and the bottom row shows density--velocity diagrams. Data from the three lanes of the A3 motorway segment are displayed, with theoretical fits superimposed.}
\label{fig:realdata_germany}
\end{figure}

\paragraph{Data-driven initialization}


\begin{figure}[htbp]
\centering

\begin{subfigure}[htbp]{\textwidth}
\centering
\includegraphics[width=\textwidth]{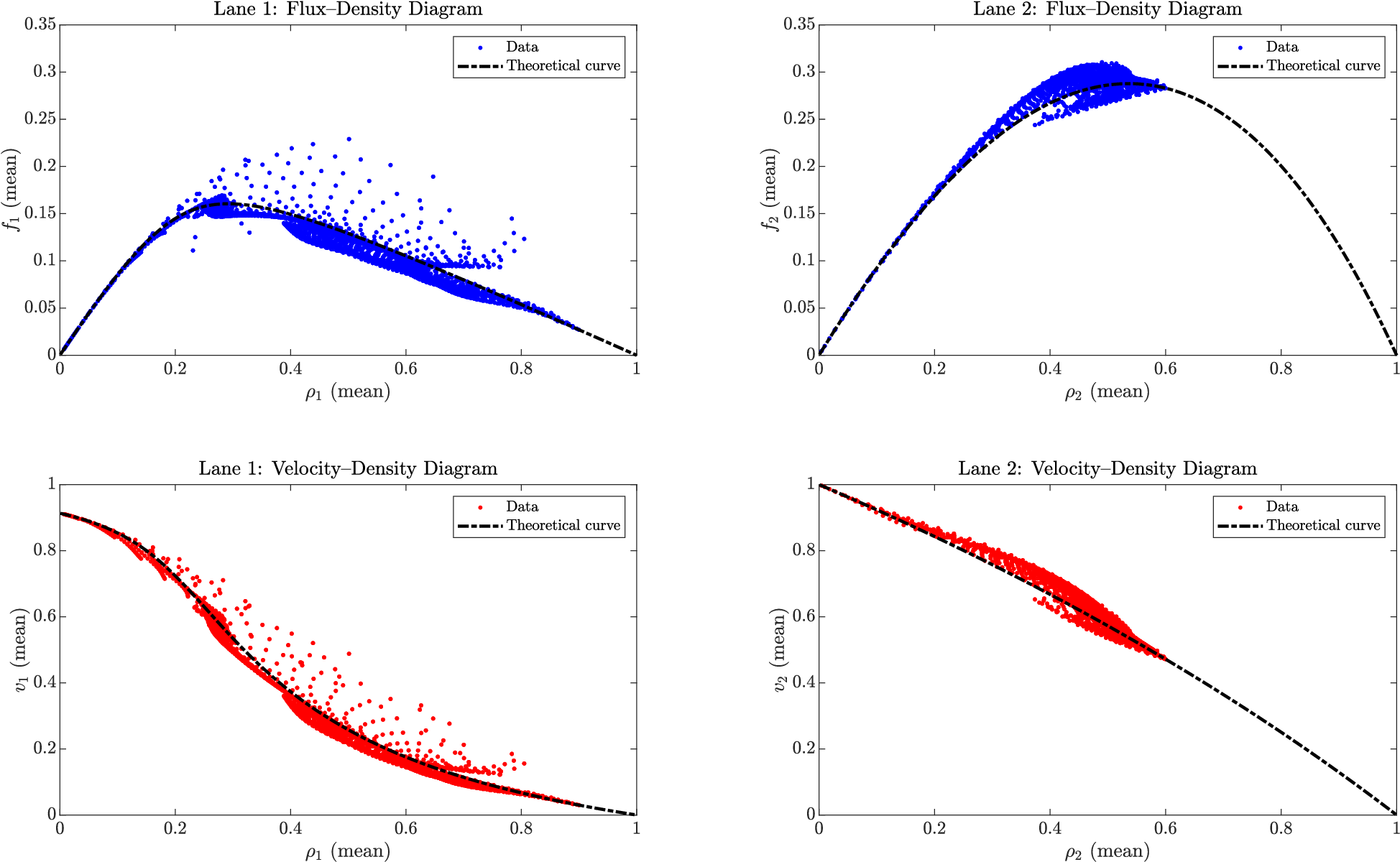}
\caption{Simulation results with Model 1. Normalized theoretical fits are compared with simulated data.}
\label{fig:mod1_italy}
\end{subfigure}

\vspace{1em}

\begin{subfigure}[htbp]{\textwidth}
\centering
\includegraphics[width=\textwidth]{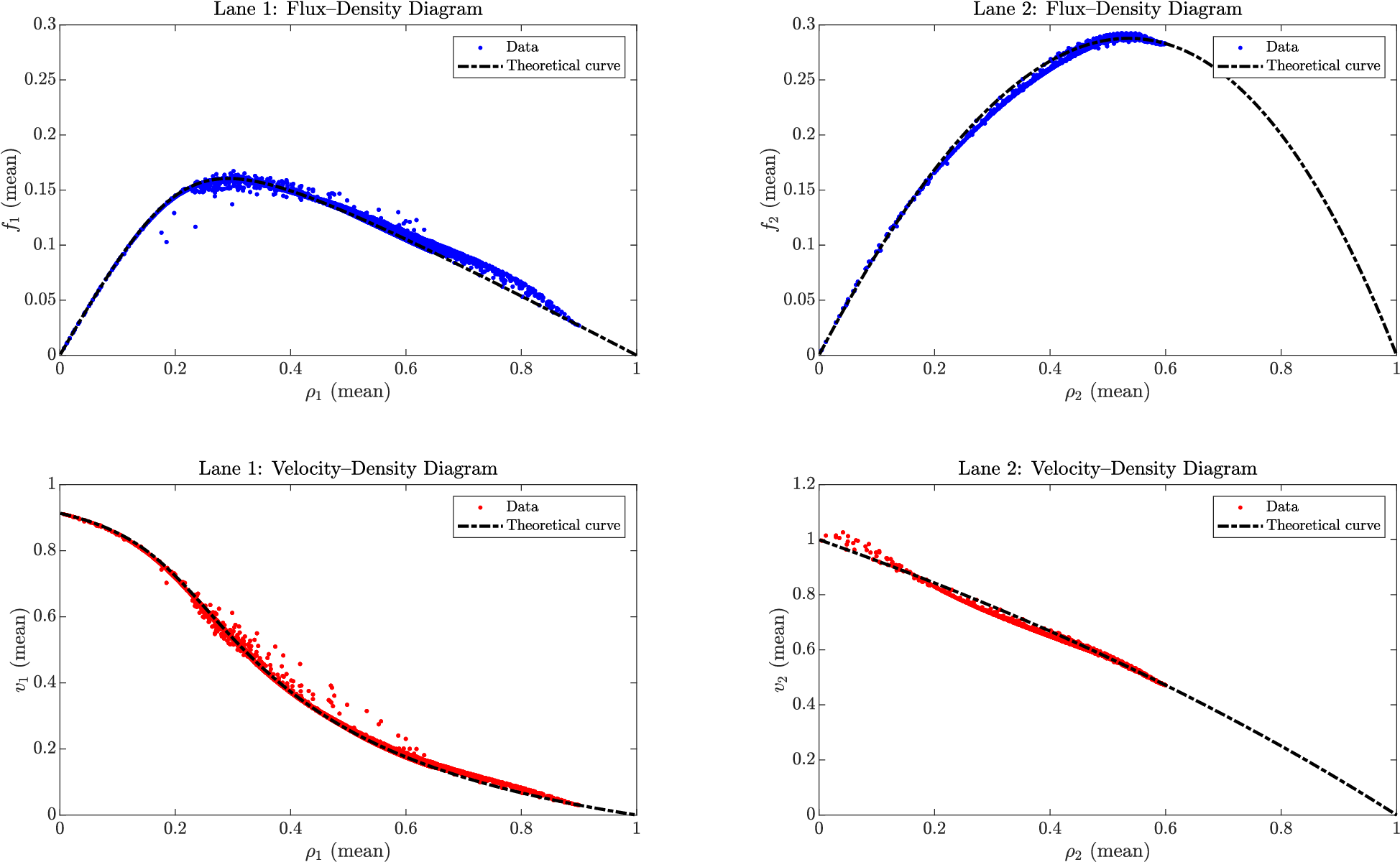}
\caption{Simulation results with Model 2. Normalized theoretical fits are compared with simulated data.}
\label{fig:mod2_italy}
\end{subfigure}

\caption{Test 5. Comparison between simulation results from Model 1 and Model 2 for the Italian two-lane dataset.}
\label{fig:italy_models}
\end{figure}

\begin{figure}[htbp]
\centering

\begin{subfigure}[htbp]{\textwidth}
\centering
\includegraphics[width=\textwidth]{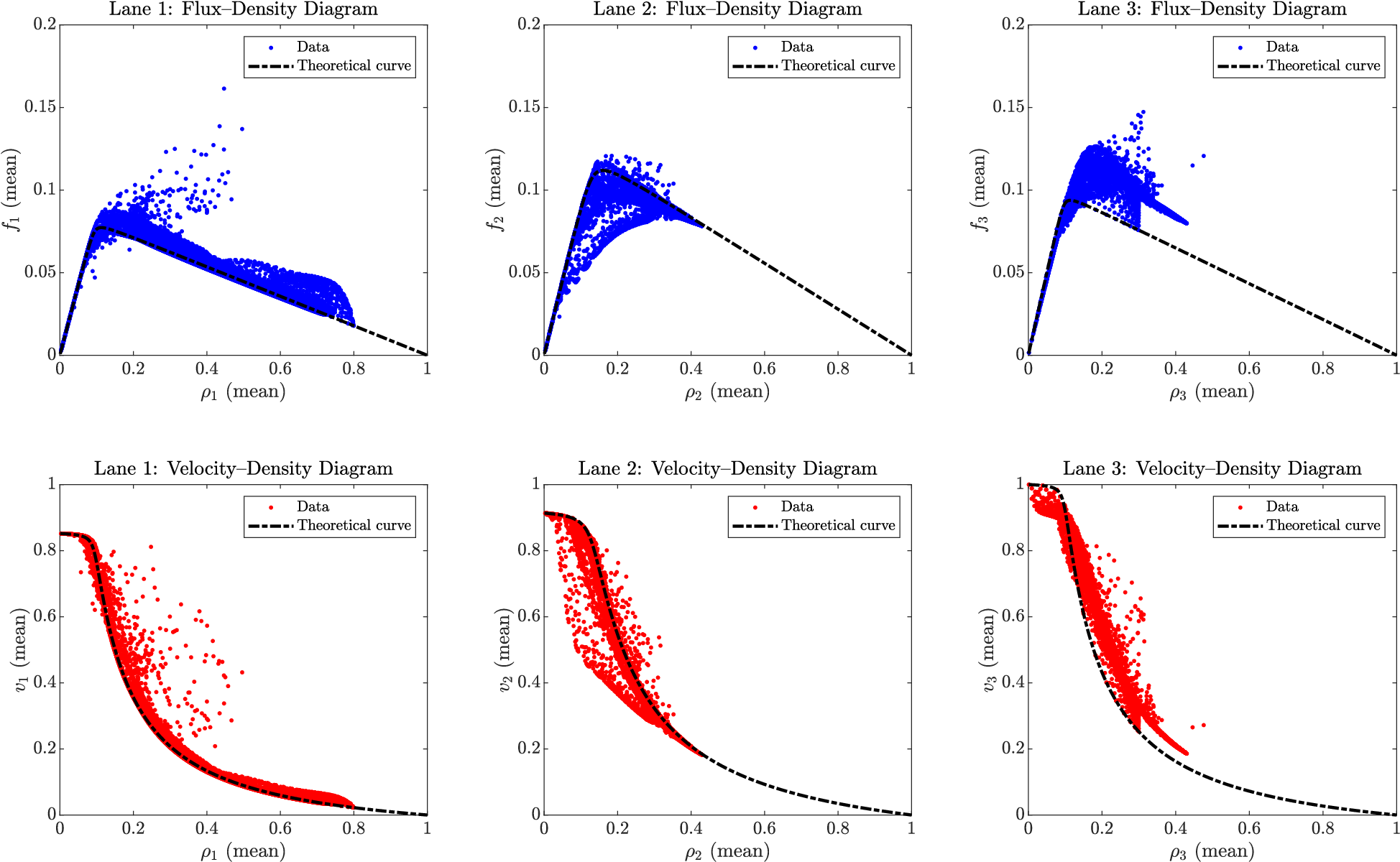}
\caption{Simulation results with Model 1. Normalized theoretical fits are compared with simulated data.}
\label{fig:mod1_germany}
\end{subfigure}

\vspace{1em}

\begin{subfigure}[htbp]{\textwidth}
\centering
\includegraphics[width=\textwidth]{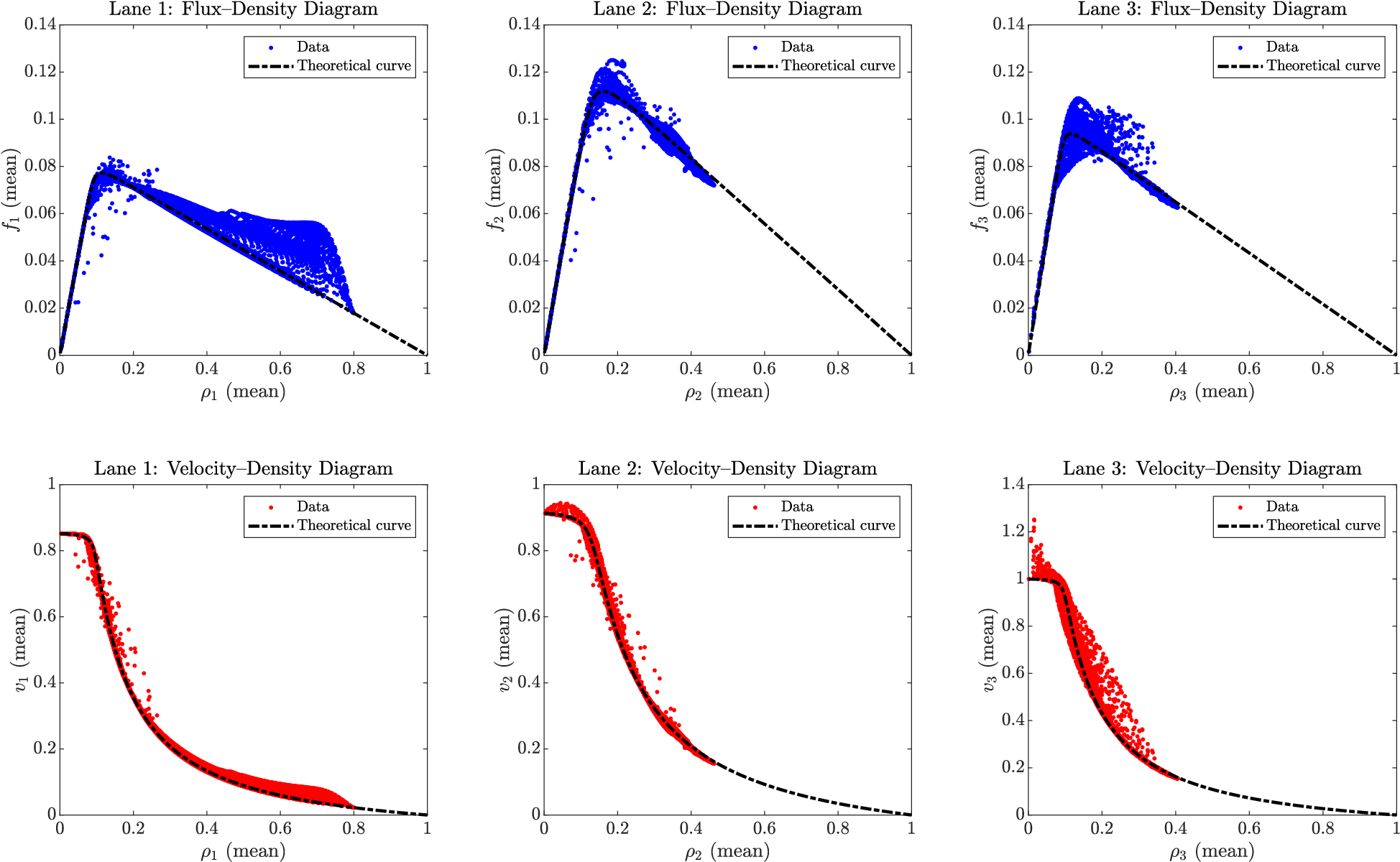}
\caption{Simulation results with Model 2. Normalized theoretical fits are compared with simulated data.}
\label{fig:mod2_germany}
\end{subfigure}

\caption{Test 5. Comparison between simulation results from Model 1 and Model 2 for the German three-lane dataset.}
\label{fig:germania_models}
\end{figure}

To further assess the consistency of our model with real traffic behavior, we performed numerical simulations inspired by the Italian and German datasets. In each case, the desired velocity profile was defined as
\begin{equation*}
V_{\alpha,\lambda,p}(\rho)=\begin{cases}   
\frac{f_{\alpha,\lambda,p}(\rho)}{\rho}, & \text{if } \rho > 0, \\
f_{\alpha,\lambda,p}'(0), & \text{if } \rho = 0,
\end{cases}
\end{equation*}
where the function \( f_{\alpha,\lambda,p} \) was specified separately for each lane according to the parameters reported in Table~\ref{tab:params_sub}. Subsequently, the functions were rescaled by normalizing both density and velocity values. Specifically, density was scaled by setting \(\rho^{\max} = 1\), while velocity was normalized with respect to the maximum value at zero density of the fastest lane, corresponding to lane 2 in the Italian two-lane setup and lane 3 in the German three-lane case.

\paragraph{Simulation setup based on Italian and German datasets} In the following, we describe the simulations for both datasets, detailing the initial conditions and model parameters, which were selected based on the analysis of the corresponding empirical data. We consider two separate numerical experiments inspired by the empirical datasets. In the first case, a two-lane configuration ($J=2$) is defined on the spatial domain $[-10, 10]$ with periodic boundary conditions. The domain is discretized with spatial step $\Delta x = 0.005$, and the simulation is run over a temporal horizon $[0,5]$. Model parameters are fixed as $\alpha=1$, $\beta=\nu=\gamma=1$, and $\eta=0.01$. The initial conditions are designed to induce lane-changing activity in both directions through alternating density profiles:
\begin{equation*}
\begin{array}{ll}
\rho_1(x,0) = 0.9\, |\sin(\pi x/10)|, & v_1(x,0) = V_1(\rho_1(x,0)), \\[0.5em]
\rho_2(x,0) = 0.6\, |\cos(\pi x/10)|, & v_2(x,0) = V_2(\rho_2(x,0)). \\
\end{array}
\end{equation*} In the second case, we consider a three-lane setup ($J=3$) on the same spatial domain, with identical discretization and time horizon. The parameters are set to $\alpha = 10$, $\beta = 2$, $\nu = \gamma = 1$, and $\eta = 0.01$. Initial conditions are constructed to trigger multidirectional lane transitions:
\begin{equation*}
\begin{array}{ll}
\rho_1(x,0) = 0.8\, |\cos(\pi x/20)|, & v_1(x,0) = V_1(\rho_1(x,0)), \\[0.5em]
\rho_2(x,0) = 0.4\, |\sin(\pi x/10)|, & v_2(x,0) = V_2(\rho_2(x,0)), \\[0.5em]
\rho_3(x,0) = 0.3\, |\cos(\pi x/15)|, & v_3(x,0) = V_3(\rho_3(x,0)). \\
\end{array}
\end{equation*}

\paragraph{Post-processing and qualitative comparison}
To facilitate visual comparison with empirical diagrams, the simulated data have been post-processed via a spatial moving average applied separately to each lane. Specifically, for each time step, the density and velocity profiles on each lane were convolved with a local averaging kernel over a fixed spatial window. This operation smoothes out high-frequency numerical oscillations and mimics the coarse-graining effect typical of detector-based measurements, where values are aggregated over space and time. As a result, the reconstructed flow–density patterns more closely resemble empirical observations and allow for a clearer qualitative assessment of the model's predictive capabilities.

The simulation results for the Italian dataset are shown in Figure~\ref{fig:italy_models}, while those for the German dataset appear in Figure~\ref{fig:germania_models}. In both cases, the simulations qualitatively reproduce key features of real-world traffic. For low densities, both the simulated and empirical diagrams follow closely the theoretical fundamental diagrams. In the medium-to-high density regime, although the amount of available data is limited, especially in the congested phase, the models correctly capture the critical density, i.e., the point at which the flow reaches its maximum and starts to decline. This behavior aligns well with the expected turning point in the fundamental diagram and reflects a realistic transition between free-flow and congested traffic conditions.

Notably, our model captures the onset of scattering in high-density regimes, a feature typically missed by first-order models that rely on a single-valued flow-density relationship. The second-order structure enables the system to explore off-equilibrium regimes, such as non-steady transitions between free and congested phases, and asymmetric lane-changing dynamics that vary with local conditions.

In addition, the simulations replicate several empirical phenomena. We note lane asymmetries caused by heterogeneous traffic composition and lane-specific preferences (e.g., faster drivers in the left lane and slower/heavier vehicles on the right). These emerge naturally from the interaction terms and the lane-dependent desired speed profiles. Moreover, the capacity drop at the onset of congestion  clearly emerges.

Based on the simulation results for both the Italian and German datasets, we observe that Model 2 exhibits superior predictive performance compared to Model 1. In particular, the diagrams produced by Model 2 show a better alignment with the empirical scattering patterns in medium and high-density regimes, as well as more realistic lane-specific differences in flow and speed. This suggests that the additional structural elements introduced in Model 2, especially in terms of lane-changing dynamics and relaxation behavior, contribute significantly to its ability to reproduce complex traffic phenomena captured in the data.

These results provide strong qualitative validation of the model's ability to replicate realistic multilane traffic patterns. Remarkably, this agreement is achieved without a full-scale parameter calibration or data assimilation procedure, suggesting that the model structure is robust and capable of generalizing across different traffic environments.

\section{Conclusions} \label{sez:conclusions}
This work investigates the derivation of second-order macroscopic models for multilane vehicular traffic starting from the microscopic Bando–Follow-the-Leader (BFtL) framework. The focus is to describe the transition from microscopic to macroscopic representations in the presence of lane-changing dynamics.

The analysis initially considers a microscopic formulation of lane-changing based on position and velocity variables. As highlighted during the transition from discrete to continuous variables, the resulting macroscopic dynamics inherently depend on the chosen closure relations. This approach yields the first macroscopic model (Model 1), closed by deriving an evolution equation directly for the velocity (requiring specific assumptions on post-transition speeds), where the coupling between lanes appears also in the convective part of the velocity equation. Due to its structural complexity, a simplified alternative (Model 2) is derived by changing the closure variable and modeling lane-changing in terms of the second conserved quantity $y$ of the ARZ model. Although the conservation of $y$ during lane changes relies on heuristic macroscopic considerations rather than strict physical principles, this results in a more compact formulation, where lane-changing contributes only through source terms in the velocity equation, while the hyperbolic part remains unchanged.

The structural divergence between the two models is also evident in their relaxation limits. Since lane-changing in Model 1 directly affects the equilibrium velocity, its relaxation yields a novel, non-standard first-order macroscopic system.Conversely, Model 2 naturally relaxes to the first-order multilane model \eqref{macro_1_ord},  derived in \cite{piu2}.

Both models are shown to be hyperbolic, ensuring well-posedness and consistency with the underlying physical principles. The anisotropic structure of the models is also analyzed.

The resulting models can be interpreted as macroscopic limits of the multilane BFtL model or as multilane generalizations of the second-order ARZ model. Importantly, the source terms are not introduced heuristically, but are rigorously derived from microscopic dynamics.

Through the numerical tests, we confirm the asymptotic consistency of the solutions of the microscopic multilane model with those of the two macroscopic multilane models. Furthermore, by calibrating the macroscopic models using experimental data, we construct simulated fundamental diagrams which are compared with empirical ones. These comparisons qualitatively validate the models' ability to capture key traffic features, with Model 2 showing a better fit with data  across different scenarios, proving that its heuristic structure ultimately provides a more accurate macroscopic description of real-world traffic.

The proposed framework provides a structured methodology for bridging microscopic and macroscopic traffic descriptions in multilane settings. The numerical results suggest that the inclusion of second-order dynamics improves the model's ability to reproduce key traffic phenomena, such as shock waves and the capacity drop.

Future work will focus on extending the framework to account for heterogeneous vehicle classes, adaptive lane-changing strategies, and and further calibration against a broader set of empirical data. These include, for instance, multivalued fundamental diagrams (which cannot be reproduced by first-order models) and the anisotropic occupancy of lanes (which cannot be captured by single-lane models) Additionally, a dedicated numerical and analytical investigation of the novel non-standard first-order system obtained from the relaxation of Model 1 will be the subject of future research.

\section*{Acknowledgments}
We are particularly grateful to Prof.~Michael Herty of RWTH Aachen University for illuminating discussions and suggestions.

M.P. acknowledges MUR-PRIN Project 2022 No. 2022N9BM3N ``Efficient numerical schemes and optimal control methods for time-dependent partial differential equations'' financed by the European Union - Next Generation EU.

G.V. acknowledges the support of MUR (Ministry of University and Research) under the MUR-PRIN PNRR Project 2022 No. P2022JC95T ``Data-driven discovery and control of multi-scale interacting artificial agent system''.

G.P. acknowledges support from Centro Nazionale Sustainable Mobility Center - PNRR-CN4\_SPOKE\_9 - B83C22002900007.

M.P., G.P., G.V. are members of the GNCS-INdAM Group.

\bibliographystyle{plain} 
\bibliography{references_new2}

\end{document}